\documentclass[12pt]{article}
\usepackage[margin=1in]{geometry} 
\usepackage{amsmath,amsthm,amssymb,todonotes}
\usepackage{multirow}
\usepackage{times}
\usepackage{listings}
\usepackage{graphicx}
\usepackage{enumerate}
\usepackage{mathrsfs}
\usepackage[colorlinks,
linkcolor=blue,
anchorcolor=blue,
citecolor=blue
]{hyperref}
\usepackage{dsfont}
\usepackage{float}
\usepackage[title]{appendix}
\usepackage{booktabs}
\usepackage{caption}
\usepackage{tikz}
\usetikzlibrary{decorations.pathreplacing}

\usepackage{chngcntr}

\usepackage{natbib}

\usepackage{setspace}
\newcommand{\R}{\mathbb{R}}

\makeatletter

\newcommand{\Rmnum}[1]{\expandafter\@slowromancap\romannumeral #1@}
\makeatother
\allowdisplaybreaks

\newtheorem{theorem}{Theorem}[section]

\newtheorem{lemma}[theorem]{Lemma}
\newtheorem{definition}[theorem]{Definition}

\newtheorem{example}[theorem]{Example}

\newtheorem{assumption}[theorem]{Assumption}

\numberwithin{equation}{section}

\title{Relative Maximum Likelihood Updating of Ambiguous Beliefs\footnote{I am grateful to Peter Klibanoff and Marciano Siniscalchi for invaluable guidance and discussions throughout completion of this paper. I thank Eddie Dekel, Itzhak Gilboa, Eran Hanany, Shaowei Ke, Yucheng Liang, Pietro Ortoleva, Rui Tang, Jingyi Xue, Chen Zhao, Mu Zhang, especially Modibo Camara and Lorenzo Stanca for insightful discussions. I thank all participants at ESWC 2020 and RUD 2020 for comments. All errors are my own.}}
\author{Xiaoyu Cheng\footnote{
		Department of Managerial Economics and Decision Sciences, Kellogg School of Management, Northwestern University, Evanston, IL, USA. E-mail: xiaoyu.cheng@kellogg.northwestern.edu}}
\begin{document}
	\maketitle
	
	\begin{abstract}
		This paper proposes and axiomatizes a new updating rule: Relative Maximum Likelihood (RML) updating for ambiguous beliefs represented by a set of priors ($C$). This rule takes the form of applying Bayes' rule to a subset of $C$. This subset is a linear contraction of $C$ towards its subset ascribing the maximal probability to the observed event. The degree of contraction captures the extent of willingness to discard priors based on likelihood when updating. Two well-known updating rules of multiple priors, full Bayesian (FB) and Maximum Likelihood (ML), are included as special cases of RML. An axiomatic characterization of conditional preferences generated by RML updating is provided when the preferences admit Maxmin Expected Utility representations. The axiomatization relies on weakening the axioms characterizing FB and ML. The axiom characterizing ML is identified for the first time in this paper, addressing a long-standing open question in the literature. \\
		
		
	\end{abstract}

	\newpage
	\section{Introduction}\label{intro}
	For decisions under uncertainty, when information is not sufficient to pin down a unique distribution over states, the decision-maker (DM)'s revealed preference sometimes is not consistent with any single probabilistic belief. Yet it could be consistent with a set of priors \citep*{Ellsberg1961,Machina1992}.  When the DM learns additional information, updating of the set of priors may involve two steps. First, she could use this information to make inference about the plausibility of each prior and discard those deemed to be implausible. Second, she applies Bayes' rule to update every prior in this refined set conditional on the information received. 
	
	The two most popular updating rules for multiple priors, \textbf{full Bayesian (FB)}\footnote{Some may also refer to it as prior-by-prior updating.} and \textbf{Maximum Likelihood (ML)}, are both examples of such an updating procedure.  FB takes the form of applying Bayes' rule to the entire set of priors. In other words, the DM does not discard any prior under FB updating. In contrast, under ML updating, the DM discards priors that do not ascribe the maximal probability (among all priors in the set) to the observed event and updates the remaining priors according to Bayes' rule. Therefore, FB and ML can be regarded as polar extremes in terms of discarding priors based on likelihood of the observed event in updating. 
	
	However, sometimes a DM may not be willing to be as extreme as either case. For example, she may find the priors assigning a very small probability to the observed event to be implausible, while at the same time would like to keep the priors ascribing an almost maximal probability. Such an intermediate behavior cannot be captured by either of the two extremes. Moreover, the preference behaviors corresponding to such non-extreme updating rules have not been extensively studied and identified in the literature.\footnote{See Section \ref{discussion} for more detailed references.}
	
	In this paper, I propose and characterize a new updating rule, \textbf{Relative Maximum Likelihood (RML)} updating, which takes the form of applying Bayes' rule to a subset of the set of priors. The specific subset is determined by a convex combination of the set of priors and its subset which ascribes a maximal probability to the observed event.\footnote{The convex combination of two sets is defined as the set of pointwise convex combinations. } It can also be understood as a linear contraction of the set of priors towards its maximum-likelihood subset. Consequently, RML is able to provide a means of intermediate updating between FB and ML. Moreover, its specific form suggests that under RML not only likelihood but also \textit{shape} of the initial set of priors are crucial for determining the set of posteriors.
	
	Formally, let a closed and convex set $C$ denote the set of priors over a state space $\Omega$. For some conditioning event $E \subseteq \Omega$, let $C^{*}(E)$ denote the subset of $C$ attaining maximum likelihood of the event $E$, i.e. $C^{*}(E) \equiv \arg\max_{p \in C}p(E)$.\footnote{These notations will be used throughout this paper.}
	
	\begin{definition}\label{rmldf}
		Under RML updating, when the event $E$ occurs, each of the priors in the set $C_{\alpha}(E)$ will be updated according to Bayes' rule. For some $\alpha \in [0,1]$, $C_{\alpha}(E)$ is defined as: 
		\begin{equation*}
			C_{\alpha}(E) \equiv  \alpha C^{*}(E) + (1-\alpha)C  =  \{\alpha p + (1-\alpha)q: p\in C^{*}(E) \text{ and } q  \in C  \}.
		\end{equation*}
	\end{definition}
	
	RML captures the intermediate updating behaviors in the sense that for all $\alpha \in [0,1]$ one has $C^{*}(E) \subseteq C_{\alpha}(E) \subseteq C$. Furthermore, FB and ML are included as extreme special cases.  
	
	The updating rule, RML, is defined under a single event $E$. For conditional preferences, a DM may apply RML updating with different parameters $\alpha$ for different events. Such a phenomenon can be seen in an experiment by \cite{liang2019}.\footnote{See Section \ref{signal} for more details.} To allow for such a possibility, I consider two different representations of conditional preferences that differ in the extent to which the parameter $\alpha$ is allowed to vary across events. On the one hand, the conditional preferences are represented by \textbf{Contingent RML} if the DM applies RML updating with parameter $\alpha[E]$ that potentially depends on the event $E$. On the other other hand, the conditional preferences generated by applying RML with a constant parameter $\alpha$ are said to be represented by \textbf{RML}. 
	
	I provide foundations for conditional preferences represented by Contingent RML and RML when the preferences admit Maxmin Expected Utility (MEU) representations \citep*{GILBOA1989141}. Under MEU, the DM evaluates prospects according to the worst expected utility generated by the set of priors. 
	
	For MEU preferences, \cite{pires2002rule} characterizes FB by a simple behavioral axiom. \cite{GILBOA199333} axiomatizes ML when preferences admit both MEU and Choquet Expected Utility (CEU) representations, which is a strict special case of the MEU preferences. Their axiomatization, however, does not extend to the more general case. In fact, there has been no axiomatization of ML under general MEU preferences in the literature. The first main result of this paper (Theorem \ref{thm1} and more generally Theorem \ref{thm5}) addresses this long-standing open question by providing a characterization of ML for MEU preferences. Moreover, this result also proves to be instrumental for the characterization of Contingent RML and RML. 
	
	Both representations are characterized by weakening the axioms leading to FB and ML. In turn, those axioms are relaxations of the well-known \textit{dynamic consistency} principle.\footnote{See \cite{ghirardato2002revisiting} for an example.} More specifically, Theorem \ref{thm2} characterizes the Contingent RML by two behavioral axioms. In addition, Theorem \ref{thm4} shows that adding one axiom to the characterization of contingent RML is necessary and sufficient to pin down a constant $\alpha$ across events. 
	
	These characterization results not only identify the behavioral foundations for the type of non-extreme updating specified by RML updating, but also suggest the common behavioral patterns under FB and ML, two seemingly orthogonal updating rules. \\
	
	Motivated by these characterizations, in this paper I further identify a key issue in the applications of ambiguity and showcase how RML can be applied to address it.
	
	Many recent applications of the MEU model assume the players update according to FB.\footnote{For example, \cite{Bose2014}, \cite{KELLNER20181} and \citet*{BEAUCHENE2019312}} In settings such as mechanism design or information design, these applications find that the introduction of ambiguity is strictly beneficial for the principal. However, there is clearly a caveat that this finding may hinge on the specific assumption of FB updating. RML, as a larger family of updating rules including FB, provides a useful tool for examining this issue. 
	
	As an illustration, Section \ref{persuasion} analyzes an example in the context of ambiguous persuasion studied by \citet*{BEAUCHENE2019312}. In this example, they construct an ambiguous persuasion scheme granting the sender strictly more payoff than the optimal Bayesian persuasion under the assumption that the receiver uses FB updating. I first show that this scheme becomes strictly worse than the optimal Bayesian persuasion whenever the receiver slightly deviates from FB in the direction of ML. The level of such a deviation can be conveniently captured by the parameter $\alpha$ under RML. This finding shows that the particular construction used by \citet*{BEAUCHENE2019312} can be non-robust with respect to updating rules. 
	
	Nonetheless, I construct an alternative ambiguous persuasion scheme. It further requires that all the probabilistic devices in an ambiguous device have the same overall probability of sending the same signal. As a result, the receiver would behave exactly the same no matter what attitude he has towards discarding priors based on likelihood. I call such a scheme \textit{uniform-likelihood ambiguous persuasion}. I further show that the sender can do strictly better using uniform-likelihood ambiguous persuasion compared to Bayesian persuasion in this example. In other words, the strict gain from using ambiguous strategies does not rely on the specific assumption of FB updating in this case.\\
	
	The remainder of this paper is organized as follows. Section \ref{pre} sets up the environment and provides a characterization of ML updating. Section \ref{rml}, the main section, gives the formal definition of preferences represented by Contingent RML and RML and also provides the foundations for them. Section \ref{signal} looks at a special environment with ambiguous signals and illustrates the predictions and interpretations that Contingent RML and RML are able to offer. Section \ref{persuasion} applies RML to the example of ambiguous persuasion and addresses the issue of robustness to updating. Section \ref{discussion} talks about the related literature. Section \ref{conclude} concludes. All the proofs are collected in Appendix \ref{ap_pf}. 
	
	\section{Preliminaries}\label{pre}
	\subsection{Set up}
	Let $\Omega$ be the set of \textit{states} with at least three elements, endowed with a sigma-algebra $\Sigma$ of \textit{events}. Denote a generic event by $E$. Let $\Delta(\Omega)$ denote the set of all finitely additive probability measures on $(\Omega, \Sigma)$ endowed with the weak* topology. Let $X$ be the set of all \textit{consequences} that are simple (i.e. finite-support) lotteries over a set of prizes $Z$ and let $x$ denote a generic element of $X$. Let $\mathcal{F}$ denote the set of \textit{simple acts}, meaning that each $f \in \mathcal{F}$ is a finite-valued $\Sigma$-measurable function from $\Omega$ to $X$. With conventional abuse of notation, denote a constant act which maps all states $\omega \in \Omega$ to $x$ simply by $x$. 
	
	The primitive is a family of preferences $\{\succsim_{E}\}_{E\in \Sigma}$ over all the acts. Let $\succsim_{\Omega} \equiv \succsim$ denote the ex-ante preference. For all the other non-empty $E\in \Sigma$, let $\succsim_{E}$ denote the conditional preference conditioning on event $E$ occurs. 
	
	First, assume that the ex-ante preference $\succsim$ admits a Maxmin Expected Utility (MEU) representation, i.e. it is represented by a closed and convex set $C \subseteq \Delta(\Omega)$ and an affine utility function $u$ such that for all $f,g \in \mathcal{F}$:  
	\begin{equation*}
		f \succsim  g  \Leftrightarrow \min\limits_{p \in C} \int_{\Omega} u(f)dp \geq  \min\limits_{p \in C}  \int_{\Omega} u(g) dp 
	\end{equation*}
	where $u(f)$ denotes a random variable $Y : \Omega \rightarrow \R$ such that $Y(\omega) = u(f(\omega))$ for all $\omega \in \Omega$. Denote the ex-ante evaluation of an act $f$ according to this MEU representation by $U(f)$. 
	
	For the ease of exposition, I impose the following assumption on the ex-ante preference: 
	
	\begin{assumption} \label{asum}(i) The set $u(X) \equiv \{ u(x) : x \in X  \}$ is unbounded from above. (ii) $C$ has finitely many extreme points.
	\end{assumption}
	
	Assumption \ref{asum} proves to be extremely useful as it helps simplify the statements of the axioms to a great extent while conveying the same intuitions. All characterizations can be achieved without these assumptions by slightly modifying the axioms.\footnote{See Appendix \ref{apx1}.} Nevertheless, these assumptions are also arguably reasonable. For example, the unbounded assumption applies to monetary payoffs in general. The finitely many extreme points assumption covers a large number of common and popular cases of multiple priors, for example the $\epsilon$-contamination. \\

	Given Assumption \ref{asum}, it is without loss of generality to normalize $u(\cdot)$ such that $u(X) = (\underline{u}, \infty)$ for some $\underline{u} \in \R_{<0}\cup\{-\infty\}$. 
	
	For each $E\in \Sigma$,  for any $f,h \in \mathcal{F}$, let $f_{E}h$ denote an act mapping all $\omega \in E$ to $f(\omega)$ and all $\omega \in E^{c}$ to $h(\omega)$. An event $E$ is strict $\succsim$-nonnull if for all $x,x' \in X$ such that $x \succ x'$, one has $x_{E}x' \succ x'$. Under MEU, an event $E$ is strict $\succsim$-nonnull if and only if $p(E) > 0$ for all $p \in C$. 
	
	For each strict $\succsim$-nonnull $E\in \Sigma$, assume that the conditional preference $\succsim_{E}$ also admits a MEU representation. It is represented by a set $C_{E} \subseteq \Delta(\Omega)$ and the same utility function $u$ such that for all $f,g \in \mathcal{F}$:  
	\begin{equation*}
		f \succsim_{E}  g  \Leftrightarrow \min\limits_{p \in C_{E}} \int_{\Omega} u(f)dp\geq  \min\limits_{p \in C_{E}}\int_{\Omega} u(g)dp
	\end{equation*}
	Meanwhile, the conditional preference $\succsim_{E}$ for all other $E$ is unrestricted.\footnote{Considering only strict $\succsim$-nonnull events is to avoid conditioning on zero-probability event under FB updating.} 
	
	Finally, assume that the conditional preferences satisfy \textit{consequentialism}: first, for all $p \in C_{E}$, $p(E) = 1$, i.e. the complement of the conditioning event is irrelevant for the conditional preference; and second, ex-ante preference and conditioning event $E$ completely determine the conditional preference $\succsim_{E}$, which rules out the possibility that the updating rule may depend on the context of decision problems (e.g. the menu of available acts). 
	
	All properties assumed for the primitive have axiomatized foundations: 
	\begin{itemize}
		\item Maxmin Expected Utility representation: axioms from \cite{GILBOA1989141}. 
		\item $u(X)$ is unbounded from above: unboundedness axiom from \citet*{Maccheroni2006}.
		\item $C$ has finitely many extreme points: no local hedging axiom from \cite{SINISCALCHI20061}. 
		\item $u$ is independent of $E$: state independence axiom from \cite{pires2002rule} or unchanged tastes axiom from \cite{hanany2007updating}. 
		\item Consequentialism: null complement axiom and the independence axioms in section 3.2 of \cite{hanany2007updating}. 
	\end{itemize}
	
	Setting up in this way allows me to establish a clear equivalence between the key axioms relating the conditional preferences to the ex-ante preference and the updating rules. Hence the key behavioral foundations characterizing different updating rules under MEU preferences can be highlighted in the representation theorems. 
	
	\subsection{Two Basic Axioms}
	For dynamic choice, the well-known dynamic consistency principle relates the conditional preferences to the ex-ante preference in two different directions (see \cite{ghirardato2002revisiting} for an example). To further distinguish between these two directions, they will be referred to separately as two basic axioms: \textbf{Contingent Reasoning (CR)} and \textbf{Dynamic Consistency (DC)} in this paper:\footnote{The statements of the axioms are equivalent to the statements of dynamic consistency adopted by \cite{ghirardato2002revisiting}.}\\
	
	\textbf{Axiom CR} (\textbf{C}ontingent \textbf{R}easoning).
	
	For all $f,h \in \mathcal{F}$ and $x\in X$, if $f\sim_{E}x$ then $f_{E}h \sim x_{E}h$.\\
	
	\textbf{Axiom DC} (\textbf{D}ynamic \textbf{C}onsistency).
	
	For all $f,h \in \mathcal{F}$ and $x\in X$, if $f_{E}h \sim x_{E}h$ then  $f\sim_{E}x$.\\
	
	Intuitively, CR emphasizes that the ex-ante preference could be recovered by considering conditional preferences contingent on whether or not the event $E$ occurs. Namely, the ex-ante comparison between the acts $f_{E}h$ and $x_{E}h$ could be done by the following procedure: Contingent on event $E$ occurring, the DM is conditionally indifferent. Meanwhile, contingent on $E$ not occurring, she receives exactly the same thing. Thus, she further concludes that she is indifferent ex-ante. 
	
	DC emphasizes the other direction, that the conditional preference should be consistent with the ex-ante preference when it does not matter what pays on the event $E^{c}$. In other words, if  two acts differ only within the event $E$, then the DM's ex-ante comparison should remain the same after learning event $E$ occurs. 
	
	As MEU preferences violate Savage's P2 (sure thing principle), whether or not the statement $f_{E}h \sim x_{E}h$ is true would also depend on the act $h$. As a result, it is no longer meaningful to state CR and DC in terms of all the acts, as a different $h$ may result in a different conclusion. Instead, the claims about CR and DC should also specify under which act it holds. 
	
	For example, a consequence $x \in X $ is said to be the \textbf{conditional certainty equivalent} of an act $f$ given event $E$ if $x \sim_{E}f$. Consider the following statement of CR, in which the act $h$ is explicitly fixed to be the conditional certainty equivalent: \\
	
	\textbf{Axiom CR-C} (\textbf{C}ontingent \textbf{R}easoning given \textbf{C}onditional certainty equivalent).
	
	For all $f \in \mathcal{F}$ and $x\in X$, if $f \sim_{E}x$, then $f_{E}x \sim x$. \\
	
	CR-C is exactly the A9 axiom in \cite{pires2002rule} and FB updating is defined in \cite{pires2002rule} by the following:
	
	\begin{definition}
		The primitive $\{ \succsim_{E}\}_{E\in \Sigma}$ is represented by \textbf{full Bayesian updating} if the following holds for all strict $\succsim$-nonnull $E \in \Sigma$ and for all $f\in \mathcal{F}$: 
		\begin{equation*}
			\min\limits_{p \in C_{E}} \int_{\Omega} u(f)dp = \min\limits_{p \in C} \int_{E} u(f)\frac{dp}{p(E)}.
		\end{equation*}
	\end{definition}
	
	Under the current framework, the representation theorem in \cite{pires2002rule} can be simplified to the following:  
	
	\begin{theorem}[Pires, 2002]
		$\{ \succsim_{E}\}_{E\in \Sigma}$ is represented by full Bayesian updating if and only if CR-C holds for all strict $\succsim$-nonnull events $E\in \Sigma$. 
	\end{theorem}
	
	This theorem highlights that CR-C is the key behavioral foundation for FB updating under MEU preferences. 
	
	\subsection{Maximum Likelihood Updating}\label{ml}
	Another popular updating rule for MEU preferences is the Maximum Likelihood (ML) updating rule proposed in \cite{GILBOA199333}. The conditional preferences represented by ML updating are defined as in the following: 
	
	\begin{definition}
		The primitive $\{\succsim_{E}\}_{E\in \Sigma}$ is represented by \textbf{maximum likelihood updating} if the following holds for all strict $\succsim$-nonnull $E \in \Sigma$ and for all $f\in \mathcal{F}$: 
		\begin{equation*}
			\min\limits_{p \in C_{E}} \int_{\Omega} u(f)dp = \min\limits_{p \in C^{*}(E)} \int_{E} u(f)\frac{dp}{p(E)}
		\end{equation*}
		where $C^{*}(E) = \arg\max_{p \in C} p(E)$.
	\end{definition}
	
	\cite{GILBOA199333} provide an axiomatic characterization of ML updating when the preferences admit both MEU and Choquet expected utility (CEU) representations. This class of preferences is a strict subset of preferences admitting MEU representations. 
	
	In that case, it is without loss of generality to let $u(X)$ be bounded and denote the best consequence in $X$ by $x^{*}$. \cite{GILBOA199333} show that ML is characterized by the following axiom:\\
	
	\textbf{Axiom CR-B} (\textbf{C}ontingent \textbf{R}easoning given \textbf{B}est consequence).
	
	For all $f \in \mathcal{F}$ and $x\in X$, if $f \sim_{E}x$, then $f_{E}x^{*} \sim x_{E}x^{*}$. \\
	
	CR-B claims that the DM's ex-ante evaluation of the act $f_{E}x^{*}$ can be recovered by applying contingent reasoning with respect to the event $E$. \cite{GILBOA199333} offer a ``pessimistic'' interpretation of this behavior: the DM's conditional preference of an act $f$ comes from the consideration that the best consequence would have been received had the complement event happened. In other words, her conditional evaluation of the act $f$ reflects a disappointment that event $E$ actually occurs. Under MEU preferences, the best consequence is ex-ante believed to be received only with the minimum probability. Consequently, the disappointment reflected in the conditional preference is in accordance with such an ex-ante belief. 
	
	As one enlarges the domain of preferences to the more general case of MEU, although this intuition may still seem to be applicable, the exact intuition captured by CR-B becomes imprecise. The following example illustrates a case where the ex-ante preference admits MEU but not CEU representation, and CR-B is violated under ML updating: 
	\begin{example}\label{exp3}
		Consider a state space $\Omega = \{\omega_{1}, \omega_{2}, \omega_{3}\}$, denote any $p \in \Delta(\Omega)$ by a vector with three coordinates $p = (p_{1}, p_{2}, p_{3})$ such that $p_{1} + p_{2} + p_{3} = 1$. Let set $C \subseteq \Delta(\Omega)$ be the closed convex hull of the following three points: $(1/2, 0, 1/2), (0, 1/2, 1/2)$ and $(1/3, 1/3, 1/3)$. 
		
		Let $X = [0,1]$ with utility function $u(x) = x$. Consider an act $f$ that pays one at $\omega_{1}$, zero at $\omega_{2}$ and is undetermined on state $\omega_{3}$. Let the conditioning event $E$ be the following set:  $\{\omega_{1}, \omega_{2}\}$. 
	\end{example}
	
	First, one can verify that a MEU preference represented by this $C$ does not admit a CEU representation.\footnote{This $C$ is not a core of any convex capacity.} If the conditional preference is given by ML updating, as $C^{*}(E)$ contains only the extreme point $(1/3,1/3,1/3)$, it implies that
	\begin{equation*}
		f \sim_{E} 1/2,
	\end{equation*}
	CR-B claims that $f_{E}x^{*}$ should be ex-ante indifferent to $1/2 _{E} x^{*}$ for $x^{*} = 1$. However, it is actually the case that 
	\begin{equation*}
		f_{E}1 \prec 1/2_{E}1,
	\end{equation*}
	i.e. CR-B is false. 
	
	One might notice that the breakdown of CR-B results from the fact that $f_{E}1$ and $1/2_{E}1$ are evaluated at different extreme points under MEU. In particular, the act $f_{E}1$ is not evaluated at the extreme point $(1/3, 1/3, 1/3)$. Thus, the DM's ex-ante evaluation of the act $f_{E}1$ is not given by the prior attaining the maximum likelihood of the event $E$. Moreover, the graphical illustration given in Figure \ref{fig:2} suggests that it is exactly because the consequence $x^{*} = 1$ is not good enough. 
	
	\begin{figure}[h]
		\centering
		\begin{tikzpicture}[x = 1.2cm, y = 1.2cm]
			\draw[black] (0,0) -- (5,0) -- (2.5, 4.33) --cycle;
			\node[right] at (5,0) {\small$\omega_{2}$}; 
			\node[right] at (2.5, 4.33) {\small$\omega_{3}$};
			\node[left] at (0,0) {\small$\omega_{1}$};
			\filldraw[blue!20]  (1.25, 2.165) -- (3.75, 2.165) -- (2.5,1.443) --cycle;
			\draw[black] (1.25, 2.165) -- (3.75, 2.165) -- (2.5,1.443) --cycle;
			\node[left] at (1.25, 2.165) {\small$(1/2,0,1/2)$};
			\node[right,xshift = 4pt] at (3.75, 2.165) {\small$(0,1/2,1/2)$};
			\node[below] at (2.5,1.2)  {\small$(1/3,1/3,1/3)$};
			\draw (2.31, -0.33) -- (5,4.33);
			\node[right] at (4.9,4.1568) {$f_{\omega_{1}, \omega_{2}}1$};
			\draw (1, 0.5776) -- (5.5, 3.1768);
			\node [right,xshift = 2pt] at (5.3,3.06128) {$f_{\omega_{1}, \omega_{2}}2$};
			\draw [blue, ->] (4.7,3.7) -- (5.2,3.1);
		\end{tikzpicture}
		\caption{Graphical Illustration of Example \ref{exp3}}
		\label{fig:2}
	\end{figure}
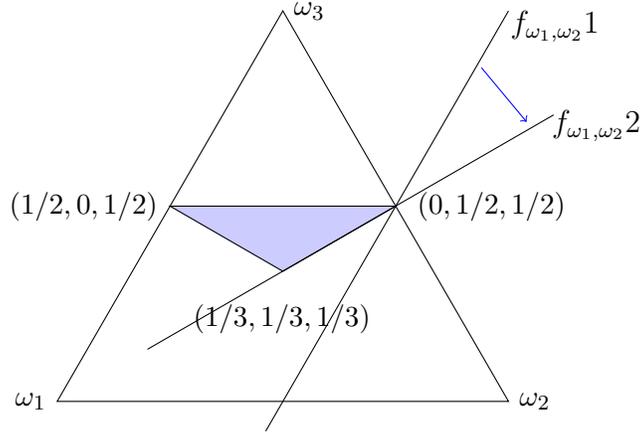
	
	In Figure \ref{fig:2}, the arrow indicates how the indifference curve of act $f_{\{\omega_{1}, \omega_{2}\}}x$ under expected utility would be changing (in angle) if one increases the value of consequence $x$. When $x = 1$, the act $f_{\{\omega_{1}, \omega_{2}\}}1$ is evaluated at the extreme point $(0, 1/2, 1/2)$ according to MEU. When $x = 2$, the act $f_{\{\omega_{1}, \omega_{2}\}}2$ is evaluated at the segment between the two extreme points $(1/3, 1/3, 1/3)$ and $(0, 1/2, 1/2)$. If one keeps increasing $x$, apparently for all $x \geq 2$, the act $f_{\{\omega_{1}, \omega_{2}\}}x$ will be evaluated at the extreme point $(1/3, 1/3, 1/3)$. 
	
	In other words, for this act $f$ and conditioning event $E $, the consequence $x$ paid on the event $E^{c}$ needs to be at least better than $x = 2$ for the ex-ante evaluation of it to be given by the prior $(1/3,1/3,1/3)$. Moreover, if $f$ instead pays two at $\omega_{1}$, then the threshold becomes $x = 4$. CR-B relies on an intuition that all the acts should have the same threshold. This example shows that such an intuition no longer holds under the general MEU preferences. 
	
	More specifically, under MEU, the ``best consequence'' in fact depends on acts and events. The following definition gives an exact characterization of the threshold for consequences that are \textit{sufficiently good}, instead of best for an act $f$ under event $E$. 
	
	\begin{definition}
		Fix any strict $\succsim$-nonnull event $E \in \Sigma$ and any act $f\in \mathcal{F}$. A consequence $\bar{x}_{E,f}$ is said to be the \textbf{threshold for sufficiently good consequences} if	for all $x^{*} \succsim \bar{x}_{E,f}$ and $x \precsim \bar{x}_{E,f}$, the following holds for all $\lambda \in [0,1]$ and constant acts $y$,$z$ with $y \prec z$: 
		\begin{equation*}
			f_{E}x^{*} \sim x_{E}x^{*} \Leftrightarrow \lambda f_{E}x^{*} + (1-\lambda)y_{E}z \sim \lambda x_{E}x^{*} +(1-\lambda) y_{E}z
		\end{equation*}
	\end{definition}
	
	To understand this definition, first notice that because $x \precsim x^{*}$ and $y \prec z$, evaluating the acts $x_{E}x^{*}$ and $y_{E}z$ under MEU criterion is given by assigning a maximal probability to the event $E$. This behavior exactly reflects the type of pessimism captured by the axiom CR-B. Now if evaluating the act $f_{E}x^{*}$ also reflects such a pessimism, i.e. by assigning a maximal probability to the event $E$, then for acts given by convex combinations of $f_{E}x^{*}$ and $y_{E}z$, they will also be evaluated under the same type of pessimism. The same also applies to convex combinations of $x_{E}x^{*}$ and $y_{E}z$, thus the independence relation. Therefore, this definition characterizes a threshold for the consequences to trigger such a pessimism for act $f$ under event $E$. 

	The existence of such a threshold for every act and event is guaranteed by Assumption \ref{asum}. In addition, any consequence $x^{*} \succsim \bar{x}_{E,f}$ is also by definition a threshold for sufficiently good consequences for the same event and act. Furthermore, the following axiom and representation theorem show that this is exactly the generalization needed to characterize ML updating under the general MEU preferences. \\
	
	\textbf{Axiom CR-S} (\textbf{C}ontingent \textbf{R}easoning given \textbf{S}ufficiently good consequence).
	
	For all $f \in \mathcal{F}$ and $x \in X$, if $f \sim_{E}x$, then $f_{E}\bar{x}_{E,f} \sim x_{E}\bar{x}_{E,f}$.  
	
	\begin{theorem}\label{thm1}
		$\{\succsim_{E}\}_{E\in \Sigma}$ is represented by maximum likelihood updating if and only if CR-S holds for all strict $\succsim$-nonnull events $E\in \Sigma$.  
	\end{theorem}
	
	In Appendix \ref{apx1}, I provide characterizations of maximum likelihood updating without Assumption \ref{asum}. 
	
	\section{Relative Maximum Likelihood Updating}\label{rml}
	\subsection{Representations}
	Recall Definition \ref{rmldf}. Under Relative Maximum Likelihood updating, the DM applies Bayes' rule to the following subset of the set of priors $C$ for some $\alpha \in [0,1]$: 
	\begin{equation*}
		C_{\alpha}(E) \equiv \alpha C^{*}(E) + (1-\alpha)C  =  \{\alpha p + (1-\alpha)q : p \in C^{*}(E) \text{ and } q \in C \}
	\end{equation*}
	
	The parameter $\alpha$ in RML updating  represents an extent of inclination towards Maximum Likelihood relative to Full Bayesian. Thus, it is possible that a DM may not have the same inclination across all the events. In contrast, such a flexibility is impossible under either FB or ML, as the parameter in those cases, by definition, are fixed to be a constant $(\alpha \equiv 0$ or $1$). 	
	
	I provide two definitions for conditional preferences generated by RML updating that differ in the extent to which the parameter $\alpha$ is allowed to vary across events. The weaker definition captures the case in which the conditional preferences are updated by RML with parameter $\alpha[E]$. In this case, $\alpha[E]$ is potentially dependent on the observed event $E$ capturing the different inclinations across events. Such conditional preferences are defined to be represented by the Contingent RML: 
	
	\begin{definition}
		The primitive $\{\succsim_{E}\}_{E\in \Sigma}$ is represented by \textbf{contingent relative maximum likelihood updating} if for all strict $\succsim$-nonnull $E\in \Sigma$ there exists $\alpha[E] \in [0,1]$ such that for all $f\in \mathcal{F}$: 
		\begin{equation*}
			\min\limits_{p \in C_{E}} \int_{\Omega} u(f)dp = \min\limits_{p \in C_{\alpha[E]}(E)} \int_{E} u(f)\frac{dp}{p(E)}
		\end{equation*}
		where $C_{\alpha[E]}(E) \equiv \alpha[E] C^{*}(E) + (1-\alpha[E])C $.
	\end{definition}
	
	By definition, Contingent RML does not require a DM to have a consistent updating rule across events. Such inconsistency may be caused by the different maximal or minimal probability of events, different contexts, or even just different labels of events. Although one can argue whether this kind of consistency in updating is desirable or not, by offering such flexibility Contingent RML is able to accommodate a broader class of behaviors. Some experimental evidence reflects exactly such inconsistency.\footnote{See discussion in Section \ref{signal}.}
	
	Nevertheless, undoubtedly, Contingent RML may not be strong enough in applications when sharper predictions of a DM's updating behaviors are needed. For example, in the case where a DM's conditional preference is observed under only one of the events. Without a constant parameter $\alpha$, the analyst cannot draw any meaningful conclusion about this DM's updating behavior under the other events. 
	
	The stronger definition captures the case in which the conditional preferences are updated by applying RML updating with the same parameter $\alpha$ across all the events. Such conditional preferences will be defined to be represented by RML: 
	\begin{definition}
		The primitive $\{\succsim_{E}\}_{E\in \Sigma}$ is represented by \textbf{relative maximum likelihood updating} if there exists $\alpha\in [0,1]$ such that for all strict $\succsim$-nonnull $E\in \Sigma$ and for all $f\in \mathcal{F}$: 
		\begin{equation*}
			\min\limits_{p \in C_{E}} \int_{\Omega} u(f)dp = \min\limits_{p \in C_{\alpha}(E)} \int_{E} u(f)\frac{dp}{p(E)}
		\end{equation*}
		where $C_{\alpha}(E) \equiv \alpha C^{*}(E) + (1-\alpha)C $.
	\end{definition}
	
	\subsection{Preference Foundation}
	RML updating is a more general class of updating rules that includes both FB and ML as special cases. Hence, the set of conditional preferences generated by RML updating is also a superset of the conditional preferences admitting FB or ML. Consequently, the axioms characterizing FB or ML may sometimes be violated by this larger set of conditional preferences. The characterization provided in this paper helps pin down the exact relaxations of those axioms to accommodate these behaviors. 
	
	Recall that FB is characterized by Contingent Reasoning given Conditional certainty equivalent ($f \sim_{E}x $ implies $f_{E}x \sim x$). Meanwhile, it is straightforward that FB can also be characterized by Dynamic Consistency given Conditional certainty equivalent : $f_{E}x \sim x$ implies $f \sim_{E}x$ (DC-C). Moreover, DC-C can be equivalently written as:\footnote{To see why they are equivalent, the necessity of the original statement of DC-C is immediate by letting $f$ to be a constant act, i.e. $f = x$. For sufficiency, first notice that for each $f$ there exists $x_{f}$ such that $f_{E}x_{f} \sim x_{f}$ holds: there exists $\underline{x}$ and $\overline{x}$ such that $f_{E}\underline{x} \succsim \underline{x}$ and $f_{E}\overline{x} \prec \overline{x}$; meanwhile $U(f_{E}x)$ is continuous and increasing in $x$. Then the original DC-C implies $f \sim_{E} x_{f}$. For any $g$ with $f_{E}x_{f} \sim g_{E}x_{f} \sim x_{f}$, the last indifference further implies $g \sim_{E} x_{f}$ by the original DC-C. }\\ 
	
	\textbf{Axiom DC-C} (\textbf{D}ynamic \textbf{C}onsistency given \textbf{C}onditional certainty equivalent). 
	
	For all $f,g \in \mathcal{F}$ and for all $x\in X$ with $x \sim_{E} f$, if  $f_{E}x \sim g_{E}x$, then $f \sim_{E}g$. \\
	
	Notice that the quantifier $x \sim_{E} f$ is simply a definition of the constant act $x$ needed for the statement of this axiom.
	
	Similarly, ML is characterized by Contingent Reasoning given Sufficiently good consequences ($f \sim_{E}x $ implies $f_{E}\bar{x}_{E,f} \sim x_{E}\bar{x}_{E,f} $ ). It can be shown to be also characterized by Dynamic Consistency given Sufficiently good consequences: $f_{E}\bar{x}_{E,f} \sim x_{E}\bar{x}_{E,f} $  implies $f \sim_{E}x $ (DC-S). 
	
	Let $\bar{x}_{E,f,g} \equiv \max\{\bar{x}_{E,f}, \bar{x}_{E,g} \}$, DC-S can also be equivalently written as: \\
	
	\textbf{Axiom DC-S}  (\textbf{D}ynamic \textbf{C}onsistency given \textbf{S}ufficiently good consequences).
	
	For all $f,g \in \mathcal{F}$, if $f_{E}\bar{x}_{E,f,g} \sim g_{E}\bar{x}_{E,f,g} $, then $f \sim_{E} g$. \\
	
	The characterization of conditional preferences generated by RML updating relies on weakening these four axioms: CR-C, CR-S, DC-C, and DC-S. On the one hand, this reliance is natural. Because RML updating includes both FB and ML as special cases, the behaviors implied by these two rules should be accommodated under RML. On the other hand, this also highlights that all the axioms characterizing RML restrict behaviors only when the act paying on the complement event $E^{c}$ is either the conditional certainty equivalent or the sufficiently good consequences.  As a result, the existing intuitions for the basic axioms can be easily adapted to understand the axioms in the present paper. 
	
	Consider the following two axioms: \\
	
	\textbf{Axiom CR-UO} (\textbf{C}ontingent \textbf{R}easoning with \textbf{U}ndershooting and \textbf{O}vershooting). 
	
	For all $f\in \mathcal{F}$ and $x \in X$, if $f \sim_{E}x$, then (i) $f_{E}x \precsim x$ and (ii) $f_{E}\bar{x}_{E,f} \succsim x_{E}\bar{x}_{E,f}$. \\
	
	\textbf{Axiom DC-CS} (\textbf{D}ynamic \textbf{C}onsistency given \textbf{C}onditional certainty equivalent and \textbf{S}ufficiently good consequences).
	
	For all $f,g \in \mathcal{F}$ and for all $x\in X$ with $x \sim_{E} f$, if  (i) $f_{E}x \sim g_{E}x$ and (ii) $f_{E}\bar{x}_{E,f, g} \sim g_{E}\bar{x}_{E,f, g}$, then $ f\sim_{E}g$. \\
	
	Notice that CR-UO is a relaxation of CR-C and CR-S, whereas DC-CS is a relaxation of DC-C and DC-S. However, the relaxations take different forms.
	
	CR-UO keeps the common premise of CR-C and CR-S ($f \sim_{E}x$) but relaxes the conclusion of the two axioms to weak inequalities. For a behavioral interpretation of the direction of those inequalities, notice $f \sim_{E} x$ says that the DM is indifferent given the additional information (event $E$). For the ex-ante comparison between $f_{E}x$ and $x$, information reduces ambiguity associated with $f_{E}x$ but does not for $x$. Thus, without information, the DM's evaluation of $f_{E}x$, with more ambiguity, would be \textbf{undershooting}, i.e. lower than $x$. On the other hand, information reduces part of the ambiguity for $f_{E}\bar{x}_{E,f}$ but all of the ambiguity for $x_{E}\bar{x}_{E,f}$. Hence, the ex-ante evaluation of $f_{E}\bar{x}_{E,f}$, with less additional ambiguity, will be \textbf{overshooting}, i.e. higher than $x_{E}\bar{x}_{E,f}$.	
	
	DC-CS keeps the common conclusion of DC-C and DC-S ($f \sim_{E} g$) but strengthens the premise by requiring the premises of both DC-C and DC-S to hold. Behaviorally, this axiom says that the DM finds herself to be conditionally indifferent between two acts only if they are sufficiently similar in the sense that are ex-ante indifferent under various comparisons. 
	
	The following theorem asserts that these relaxations are exactly the ones needed to characterize Contingent RML: 
	
	\begin{theorem}\label{thm2}
		$\{\succsim_{E}\}_{E\in \Sigma}$ is represented by contingent relative maximum likelihood updating if and only if CR-UO and DC-CS hold for all strict $\succsim$-nonnull events $E\in \Sigma$. Furthermore for every such $E$, $\alpha[E]$ is unique if there exists $p, p' \in C$ s.t. $p(E) \neq p'(E)$. 
	\end{theorem}
	
	For a high-level intuition, CR-UO is equivalent to the conditional preference is given by some intermediate set $C(E)$ with $C^{*}(E) \subseteq C(E) \subseteq C$. Given this restriction, DC-CS further pins down the set $C(E)$ to be given by the functional form of RML updating with some $\alpha[E]$. Next, I will provide a sketch of the proof of this result:
	
	\begin{proof}[Sketch of proof of Theorem \ref{thm2}] 
		
		For sufficiency of CR-UO and DC-CS, fix any strict $\succsim$-nonnull $E\in \Sigma$, the proof proceeds by the following four steps: \\
		
		\underline{Step 1.} For any $f\in \mathcal{F}$, CR-UO implies that $f_{E}\bar{x}_{E,f} \succsim x_{E}\bar{x}_{E,f} \succsim x \succsim f_{E}x$. Then it can be easily seen from Figure \ref{fig1} that there always exists  an $\alpha[E,f] \in [0,1]$  such that the following equation holds: 
		\begin{equation}\label{eq}
			\alpha[E,f] U(f_{E}\bar{x}_{E,f} ) + (1-\alpha[E,f])U(f_{E}x)  = \alpha[E,f] U(x_{E}\bar{x}_{E,f} ) + (1-\alpha[E,f]) U(x) 
		\end{equation}
		
		Furthermore, this $\alpha[E,f]$ is unique if either $f_{E}x \prec x$ or $f_{E}\bar{x}_{E,f,x}  \succ x_{E}\bar{x}_{E,f,x} $ holds. \\
		
		\begin{figure}[h]
			\centering
			\begin{tikzpicture}[x= 1cm, y = 1cm]
				\draw[thick, ->] (0,0) -- (5,0);
				\draw[thick, ->] (0,0) -- (0,5);
				\draw (0,0) -- (4.5,4.5);
				\draw (0,1) -- (4.5, 3);
				\draw [dashed] (4.5,0) -- (4.5,4.5);
				\node [below] at (4.5,0) {1};
				\node [below] at (0,0) {0};
				\node [left] at (0,0) {$U(f_{E}x)$};
				\node [left] at (0,1) {$U(x)$};
				\node [right] at (4.5,4.5) {$U(f_{E}\bar{x}_{E,f} )$};
				\node [right] at (4.5,3) {$U(x_{E}\bar{x}_{E,f} )$};
				\draw [dashed] (1.8,0) -- (1.8,1.8);
				\node [below] at (1.8,0) {$\alpha[E,f]$};
			\end{tikzpicture}
			\caption{Graphical Illustration of Step 1}
			\label{fig1}
		\end{figure}
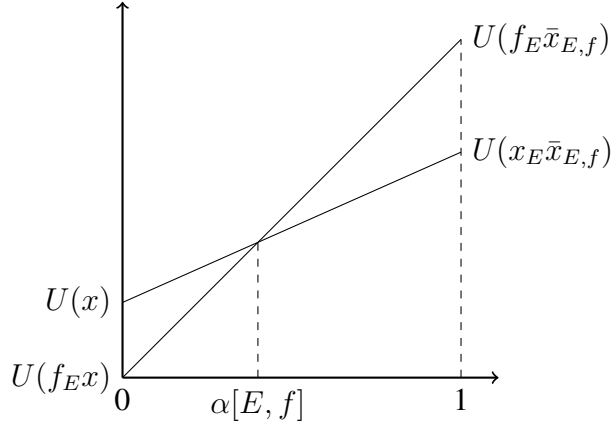
		
		\underline{Step 2.} Fix any $f\in \mathcal{F}$ such that $\alpha[E,f]$ is unique (if it does not exist then $\alpha[E]$ is not unique). First, show that for the act $f_{\lambda} y = \lambda f + (1-\lambda)y$ with any $\lambda \in (0,1]$ and $y\in X$,  one has $\alpha[E, f_{\lambda}y] = \alpha[E,f]$. Namely, $\alpha[E,f]$ is unchanged under the transformation of $f$ by taking mixtures with constant acts. This is given by equation (\ref{eq}) and the certainty independence axiom in \cite{GILBOA1989141}. 
		
		Then for any act $g\in \mathcal{F}$ that cannot be obtained by taking mixtures of $f$ with constant acts, show that there always exists such transformations of $f$ and $g$ that the pair of transformed acts satisfies the premises of DC-CS. \\
		
		\underline{Step 3.} Given the pair of transformed acts satisfying the premises of DC-CS, the axiom further implies that $\alpha[E,f]$ needs to be equal to $\alpha[E,g]$. Since the construction in step 2 works for arbitrary $g\in \mathcal{F}$, it thus concludes that $\alpha[E,f]$ needs to be a constant across all $f \in \mathcal{F}$. In other words, equation (\ref{eq}) holds with $\alpha[E]$ independent of $f$: 
		\begin{equation}\label{eq2}
			\alpha[E] U(f_{E}\bar{x}_{E,f} )+ (1-\alpha[E]) U(x) =  \alpha[E] U(x_{E}\bar{x}_{E,f} ) + (1-\alpha[E]) U(x) 
		\end{equation}
		\bigskip 
		
		\underline{Step 4.} By plugging into each term in equation (\ref{eq2}), the DM's conditional evaluation of any act $f \in \mathcal{F}$ can be shown to be represented by
		\begin{equation*}
			\min\limits_{p\in C_{\alpha[E]}(E)} \int_{E} u(f) \frac{dp}{p(E)}
		\end{equation*}
		i.e. it is represented by Contingent RML.  
	\end{proof}
	
	Theorem \ref{thm2} shows that CR-UO and DC-CS characterize conditional preferences that are given by Contingent RML where the parameter $\alpha[E]$ may be different across events. Clearly, this is because both axioms restrict behaviors only within each conditioning event $E$, yet do not explicitly restrict behaviors across events. Notice that imposing either CR-C or CR-S for all events actually also restricts behaviors across events, because the resulting updating rule takes the form $\alpha[E] \equiv 0 $ or $\alpha[E] \equiv 1$. However, the two axioms characterizing the Contingent RML do not have this additional power. \\

	To further restrict behaviors across events, let $\bar{x}_{E_{1}, E_{2}, f,g} \equiv \max\{\bar{x}_{E_{1},f,g}, \bar{x}_{E_{2},f,g} \}$ denote the sufficiently good consequence over two acts and two events. Consider the following axiom, which features a similar statement as DC-CS but across two events $E_{1}$ and $E_{2}$: \\
	
	\textbf{Axiom EC} (\textbf{E}vent \textbf{C}onsistency). 
	
	For all $f,g\in \mathcal{F}$ and $E_{1}, E_{2} \in \Sigma$ and for all $x, x_{1}^{*}, x_{2}^{*} \in X$ with $x \sim_{E_{1}} f$, $x_{1}^{*}\succsim \bar{x}_{E_{1}, E_{2}, f,g} $, $x_{2}^{*} \succsim \bar{x}_{E_{1}, E_{2}, f,g} $ and $x_{E_{1}}x_{1}^{*} \sim x_{E_{2}}x_{2}^{*}$, if  (i) $f_{E_{1}}x \sim g_{E_{2}}x$ and (ii) $ f_{E_{1}}x_{1}^{*} \sim g_{E_{2}}x_{ 2}^{*}$, then $g \sim_{E_{2}} x$. \\
	
	First, notice that EC is also a relaxation of DC-C and DC-S, because under DC-C the first premise implies $g \sim_{E_{2}} x$, and under DC-S the second premise implies $g \sim_{E_{2}} x$. It is a relaxation of DC-C and DC-S different from DC-CS. In particular, EC and DC-CS are orthogonal in the sense that they do not imply each other. 
	
	Event consistency, as suggested by its name, uses constant acts to calibrate the conditional preferences. More specifically, two acts have the same conditional certainty equivalent,  one under events $E_{1}$ and the other under $E_{2}$, if they are indifferent under the ex-ante preference when some specific act is paid on $E_{1}^{c}$ and $E_{2}^{c}$ respectively. The behavioral intuition of DC-CS can be adapted here to provide a similar explanation for this axiom. 
	
	The following theorem shows that EC achieves exactly its goal: calibrating a constant $\alpha$ to represent conditional preferences that are given by RML updating: 
	\begin{theorem}\label{thm4}
		$\{\succsim_{E}\}_{E\in \Sigma}$ is represented by relative maximum likelihood updating if and only if CR-UO, DC-CS and EC hold for all strict $\succsim$-nonnull events $E\in \Sigma$. Furthermore, $\alpha$ is unique if there exists such an $E$ that $p(E) \neq p'(E)$ for some $p,p' \in C$. 
	\end{theorem}
	
	\textbf{Remark.} The two representations, Contingent RML and RML, focus on the two extreme types of restrictions on the parameters. Some may find an intermediate restriction reasonable, for example, $\alpha[E_{1}] = \alpha[E_{2}]$ implies $\alpha[E_{1}\cup E_{2}] = 	\alpha[E_{1}]$. I leave the exploration of this problem for future research. 
	
	\section{Updating Ambiguous Signals}\label{signal}
	In this paper, the state space $\Omega$ is the grand state space, which could include both payoff-relevant states and payoff-irrelevant signals. To facilitate discussions, in this section, $\Omega$ will be explicitly written as the Cartesian product of states and signals: $\Theta \times S$ with generic element $\theta \times s$. Then the event in which signal $s$ realizes is given by $E = \Theta \times s$.
	
	Given a set of priors $C \subseteq \Delta(\Theta \times S)$, its marginal distribution over the states\footnote{I will refer to payoff-relevant states simply as states hereafter, as doing so should cause no confusion.} $\Theta$ is ambiguous if there exists $p, p' \in C$ such that $p(A\times S) \neq p'(A\times S)$ for some $A\times S \in \Sigma$. Conversely, it is probabilistic over the states if  $p( A \times S) = p'(A \times S)$ for all $A\times S \in \Sigma$ and $p, p' \in C$. In the latter case, ambiguity could arise only through the signals. This section illustrates the implications of RML updating in this special environment. 
	
	Though special, this environment has been adopted in various applications, for example, \cite{Bose2014}, \citet*{BEAUCHENE2019312}, \cite{KELLNER20181}. In addition, several recent experimental papers also focus on understanding subjects' response to ambiguous signals when there is no other source of ambiguity, for example, \cite{liang2019}, \cite{shishkin2019ambiguous}, \citet*{kellner2019reacting}. Therefore, understanding the implications of RML updating in this specific environment will be helpful not only for comparison with other updating rules, but also for illustrating the theoretical contributions of RML for future applications.  
	
	In the following, I will use a stylized example to illustrate the theoretical predictions RML updating is able to offer here. 
	
	\begin{example}\label{exp1}
		There are two payoff-relevant states: $\Theta = \{ \theta_{1},\theta_2\}$ with unambiguous marginal prior $p(\theta_{1}) = \beta \in [0,1]$. The DM is evaluating a bet $f$ that pays one at $\theta_{1}$ and nothing at $\theta_{2}$. There are two signals: $S = \{s_{1}, s_{2}\}$ and the signaling structure is ambiguous, meaning that there are two possible correlations between signals and states: 
		\begin{align*} 
			p(s_{i}| \theta_{i}) = \lambda_{1}  \\
			p(s_{i}| \theta_{i}) = \lambda_{2} 
		\end{align*}
		and which correlation generates the signal is unknown. Without loss of generality, assume $(\lambda_{1} + \lambda_{2})/2 \geq 1/2$ and $\lambda_{1} \geq \lambda_{2}$. Namely, the signal $s_{1}$ is ``on average'' an informative signal for the state $\theta_{1}$. Moreover, the first signaling device governed by $\lambda_{1}$ is more accurate than $\lambda_{2}$. 
	\end{example}
	
	Suppose the DM's preference is represented by a set of priors coinciding exactly with the information given in this example. Her set of priors $C$ can be easily parameterized by a parameter $\mu \in [0,1]$ denoting the probability that the first signaling device governed by $\lambda_{1}$ is the one generating the signals. Namely, $C = \{p_{\mu} \in \Delta(\Theta \times S) : \mu \in [0,1]\}$ where $p_{\mu}$ is defined as in the following: 
	\begin{align*}
		&p_{\mu} (\theta_{1} \times \{s_{1}, s_{2}\}) = \beta \\
		&p_{\mu} (\theta_{2} \times \{s_{1}, s_{2}\}) = 1-\beta\\
		&p_{\mu}(\theta_{1}\times s_{1}) = \beta[\mu\cdot \lambda_{1} + (1-\mu) \cdot \lambda_{2}]\\
		&p_{\mu}(\theta_{2}\times s_{1}) = (1-\beta)[\mu\cdot (1-\lambda_{1}) +(1-\mu)\cdot (1-\lambda_{2}) ]\\
		&p_{\mu} (\{\theta_{1}, \theta_{2} \}\times s_{1}) = p_{\mu}(\theta_{1}\times s_{1}) + p_{\mu}(\theta_{2}\times s_{1})
	\end{align*}
	
	To fully characterize the DM's behavior in this example, there are in total six different cases: $\beta > 1/2, \beta = 1/2$ or $\beta < 1/2$ combined with either signal $s_{1}$ or $s_{2}$. Only one of the cases will be derived here in detail, and a summary of behaviors in all these cases will be provided afterwards in Table \ref{table1}. 
	
	Consider the case $\beta > 1/2$ and signal $s_{1}$ is realized. Notice that when $\beta > 1/2$, the likelihood of signal $s_{1}$, given by $p_{\mu} (\{\theta_{1}, \theta_{2} \}\times s_{1})$, is maximized when $\mu = 1$. Moreover, thanks to the simple parametrization of the set $C$, for any $\alpha \in [0,1]$ the set $C_{\alpha}(\{\theta_{1}, \theta_{2} \}\times s_{1}) $ is given by 
	\begin{equation*}
		C_{\alpha}(\{\theta_{1}, \theta_{2} \}\times s_{1}) = \{p_{\mu} \in C: \mu \in [\alpha, 1] \}
	\end{equation*}
	Under RML, the DM updates only the priors $p_{\mu}$ for $\mu \in [\alpha, 1]$, and for each $p_{\mu}$ the posterior about states after observing $s_{1}$ is given by
	\begin{equation*}
		\pi_{\mu}(\theta_{1} | s_{1}) \equiv \frac{p_{\mu}(\theta_{1}\times s_{1})}{p_{\mu} (\{\theta_{1}, \theta_{2} \}\times s_{1})}
	\end{equation*}
	
	Therefore, the DM's set of posteriors will be the following interval\footnote{Notice that $\pi_{\mu}(\theta_{1} | s_{1})$ is increasing in $\mu$.}: $\left[ \pi_{\alpha}(\theta_{1} | s_{1}), \pi_{1}(\theta_{1} | s_{1})\right]$. Clearly, under FB, the DM's set of posteriors will be $\left[ \pi_{0}(\theta_{1} | s_{1}), \pi_{1}(\theta_{1} | s_{1})\right]$, and under ML the DM will end up with a single posterior: $\pi_{1}(\theta_{1} | s_{1})$. 
	
	Furthermore, under MEU the DM's conditional evaluation of the bet $f$ will be given by the lowest posterior of state $s_{1}$. Let $x_{f}^{FB}, x_{f}^{RML}$ and $x_{f}^{ML}$ denote the conditional certainty equivalent of the bet $f$ under FB, RML, and ML updating. Then one has $x_{f}^{FB} \leq x_{f}^{RML} \leq x_{f}^{ML}$ for all $\alpha \in [0,1]$. Especially, $x_{f}^{RML}$ is increasing with respect to $\alpha$ and coincides with $x_{f}^{FB}$ and $x_{f}^{ML}$ when $\alpha = 0$ and $\alpha = 1$ respectively. 
	
	Table \ref{table1} summarizes the predictions of behaviors under RML in all six cases. The first two rows are the only cases where the DM's evaluation of $f$ depends on the value of $\alpha$. In other words, those are the cases where different updating rules affect one's conditional evaluation of $f$. Thus, RML is able to provide richer predictions than either FB or ML. 
	
	The third and fourth rows are the cases where the MEU evaluation of $f$ is given by the posteriors updated from the maximum likelihood priors. Thus, even though the updated beliefs are different under FB and ML, MEU makes their evaluations the same. For this reason, although the updated belief under RML is also different from FB and ML for any $\alpha \in (0,1)$, the MEU evaluation under RML is still the same. 
	
	Finally, the fifth and sixth rows represent the same type of special case where every prior agrees on the likelihood of every signal. It is the case where FB coincides with ML, and therefore RML does not have an extra bite. All the priors will be updated under RML for all $\alpha \in [0,1]$, and thus the updated belief will always be the same. 
	
	\begin{table}[h]
		\centering
		\begin{tabular}{ccccc}
			\hline
			$\beta$ & Signal & ML prior & Evaluation& Comparison with\\
			&&&of $f$&probabilistic signal \\ 
			\hline 
			&             &                     &                                                    & $\alpha <1/2$: lower\\
			$>1/2$& $s_{1}$ &  $\mu = 1$ &  $\pi_{\alpha}(\theta_{1}|s_{1})$ & $\alpha = 1/2$: equal\\
			&              &                   &                                                    & $\alpha >1/2$: higher\\
			
			&             &                     &                                                    & $\alpha <1/2$: lower\\
			$>1/2 $& $s_{2}$ & $\mu = 0$ & $\pi_{\alpha}(\theta_{1}|s_{2})$  & $\alpha = 1/2$: equal\\
			&              &                   &                                                    & $\alpha >1/2$: higher\\

			$<1/2 $& $s_{1}$ & $\mu = 0$ & $\pi_{0}(\theta_{1}|s_{1})$  &  All $\alpha$: lower\\
			
			$<1/2 $& $s_{2}$ & $\mu = 1$ & $\pi_{1}(\theta_{1}|s_{2})$  & All $\alpha$: lower  \\
			
			$=1/2 $& $s_{1}$ & All $\mu$ & $\pi_{0}(\theta_{1}|s_{1})$ & All $\alpha$: lower \\ 
			
			$=1/2 $& $s_{2}$ & All $\mu$ & $\pi_{1}(\theta_{1}|s_{2})$ &  All $\alpha$: lower\\ 
			\hline 
		\end{tabular} 
		\caption{Summary of Example \ref{exp1} under RML}
		\label{table1}
	\end{table}
	
	The last column of Table \ref{table1} presents a comparison highlighted in experiments by \cite{liang2019}. It compares the DM's s conditional evaluations of $f$ under this ambiguous signaling structure with a probabilistic signaling structure where the correlation equals to $(\lambda_{1} + \lambda_{2})/2$. The latter is designed to reflect an average accuracy of the two signaling devices in the ambiguous signaling structure. For probabilistic signals, the DM simply follows Bayesian updating.  
	
	In the case where $s_{1}$ is realized, notice that the DM's posterior under the probabilistic signaling structure is exactly given by $\pi_{1/2}(\theta_{1} | s_{1})$. Therefore, for $\beta \leq 1/2$ the DM's conditional evaluation under ambiguous signal is always \textbf{lower} than her conditional evaluation under the probabilistic signal. When $\beta > 1/2$, this comparison will depend on the value of $\alpha$. More specifically,  $\alpha < 1/2$ implies that the conditional evaluation under ambiguous signal is lower, $\alpha > 1/2$ implies that the conditional evaluation under ambiguous signal is higher, and $\alpha = 1/2$ implies they are the same. Thus, from the perspective of RML, this type of comparison actually reflects the DM's attitudes towards discarding priors based on likelihood as captured by $\alpha$. 
	
	For the experimental results shown by Table 4.1 in \cite{liang2019}, each row there can be categorized into some case in Table \ref{table1} here. In his terminology, $s_{1}$ is ``good news'' and $s_{2}$ is ``bad news''. A rather interesting pattern emerging from the choice data is that when $\beta > 1/2$, the subjects' comparison is lower after receiving good news ($s_{1}$) and higher after receiving bad news ($s_{2}$). Notice that in Table \ref{table1} the two directions of the comparisons are aligned with $\alpha < 1/2$ and $\alpha > 1/2$ respectively. This implies that the subjects cannot have the same attitude towards discarding priors for all the signals. This pattern can be captured by Contingent RML, which captures exactly such flexibility in behaviors. 
	
	Moreover, the other rows of Table 4.1 correspond to the cases where $\alpha$ is irrelevant for this comparison. The choice data there is also aligned with behaviors under RML updating.\footnote{Only row 8 of Table 4.1 in \cite{liang2019} is different from the prediction of RML. The last two rows do not specify the signals and thus are ignored.}
	
	Therefore, not only does Contingent RML with $\alpha[\{\theta_{1}, \theta_{2} \}\times s_{1}]  <1/2$ and $\alpha[\{\theta_{1}, \theta_{2} \}\times s_{2}]  >1/2$ offer an interpretation for the different directions of comparison across signals, but it also provides an approach to accommodate almost all the behavioral patterns in that experiment.\footnote{The same conclusion can be drawn also from a within-subject comparison shown in Table B.1 of \cite{liang2019}}
	
	In summary, for this special case where a DM has probabilistic belief over states and ambiguous belief over signals, RML is able to provide richer predictions. In addition, such richness also proves to be useful for applications involving ambiguous signals. 
	
	\section{Uniform-likelihood Ambiguous Persuasion: An Example}\label{persuasion}
	In settings such as mechanism design or information design, many recent papers show that the designers could gain strictly more payoff by introducing ambiguity in their communications. These findings are observed primarily in the case where the agent/receiver is modeled by MEU preferences with FB updating. This section uses an example in the context of information design to explore both the robustness of such a finding once the FB assumption is relaxed to RML and how the generalization to RML may affect optimal persuasion. 
	
	Consider the persuasion game studied by \cite{Kamenica2011}, where the sender commits to a signaling device and the receiver takes actions contingent on the signal realizations. Bayesian persuasion describes the case where the available signaling devices for the sender are only probabilistic: Each device is a mapping from the state space to distributions over the signal space. 
	
	Suppose the sender also has access to an ambiguous device, which specifies a set of probabilistic devices. Moreover, which probabilistic device is used to generate the signals is unknown to both the sender and receiver.\footnote{The sender can choose a set of probabilistic devices and delegate the choice from this set to a third party or to the draw from an Ellsberg urn to make the signals ambiguous.} \citet*{BEAUCHENE2019312} (BLL henceforth) show that when the receiver's preference is represented by MEU with FB updating, the sender is able to gain strictly more payoff by using an ambiguous device compared to the optimal probabilistic device.
	
	Notice that once the sender commits to an ambiguous device, the receiver is facing a scenario of updating ambiguous signals extensively discussed in the previous section. Clearly, the FB assumption excludes the possibility that the receiver may be willing to make inference about the devices based on the likelihood of the realized signal.  This raises the possibility that the strict gain from using ambiguous signals may hinge on assumptions concerning such inference. 
	
	RML updating provides a convenient tool to address this issue. More specifically, with the parametrization of RML, one can check for which values of the parameter $\alpha$ the same conclusion holds. 
	
	To give an example of such analysis, I take the illustrative example from BLL and relax their FB assumption to RML. For this particular example, the ambiguous device they construct is more beneficial than the optimal probabilistic device only when $\alpha = 0$, i.e. only when RML reduces to FB. In other words, their finding of strict benefits from using ambiguous signals seems to crucially depend on the specific assumption of FB. 
	
	Nonetheless, in the same example, I further show that there exists a \textit{uniform-likelihood} ambiguous device. It induces the same belief and thus action from the receiver for all possible values of $\alpha$ under RML. Moreover, the sender gains strictly more payoff from this device than from the optimal probabilistic device. In other words, the strict benefit from using ambiguous signals in this example is robust to the possibility that the receiver may make likelihood-based inference from the realized signal. 
	
	\begin{example}[Illustrative Example in BLL]\label{exp4}
		There are two states $\{ \omega_{l}, \omega_{h}\}$ with a uniform marginal prior, and the receiver has three feasible actions: $\{a_{l}, a_{m}, a_{h}\}$. Payoff of the sender and receiver for each state and action is given as follows: 
		\begin{center}
			\begin{tabular}{c|cc}
				& $\omega_{l}$ & $\omega_{h}$ \\ 
				\hline 
				$a_{l}$ & (-1,3) & (-1,-1) \\ 
				$a_{m}$ & (0,2) & (0,2) \\ 
				$a_{h}$ & (1,-1) & (1,3) \\ 
			\end{tabular} 
		\end{center}
		where in each cell, the first number is the sender's payoff and the second is the receiver's. 
	\end{example}
	
	Notice the sender always prefers the receiver to take higher action, yet the receiver prefers to choose the action that matches the state. The sender optimizes her ex-ante payoff under MEU preferences. Under Bayesian persuasion, the sender's optimal payoff is given by ($u_{s}$ stands for the sender's utility function)
	\begin{equation*}
		\frac{1}{2}u_{s}(a_{m}) + \frac{1}{2}u_{s}(a_{h})
	\end{equation*}
	
	BLL shows that the sender's optimal value under ambiguous persuasion is strictly higher than this value and can be achieved by the following ambiguous device.\footnote{The construction here is equivalent to BLL but notationally simpler.} 
	
	Let $\{m_{l}, m_{l}', m_{h}\}$ be the set of signals. Consider an ambiguous device contains the following two probabilistic devices $\pi_{1}$ and $\pi_{2}$: For some $\lambda \in [0,1]$, 
	
	\begin{table}[H]
		\centering
		\begin{tabular}{c|cc}
			$\pi_{1}(m|\omega)$ & $\omega_{l}$ & $\omega_{h}$ \\ 
			\hline 
			$m_{l}$ & $2/3\lambda$ & 0 \\ 
			$m_{l}'$ & $3/4(1-\lambda)$ & $1/4(1-\lambda)$ \\ 
			$m_{h}$ & $1/3\lambda + 1/4(1-\lambda)$ & $\lambda + 3/4(1-\lambda)$ \\ 
		\end{tabular} 
	\end{table}
	
	\begin{table}[H]
		\centering
		\begin{tabular}{c|cc}
			$\pi_{2}(m|\omega)$ & $\omega_{l}$ & $\omega_{h}$ \\ 
			\hline 
			$m_{l}$ & $3/4(1-\lambda)$ & $1/4(1-\lambda)$ \\ 
			$m_{l}'$ & $2/3\lambda$ & 0 \\ 
			$m_{h}$ & $1/3\lambda + 1/4(1-\lambda)$ & $\lambda + 3/4(1-\lambda)$ \\ 
		\end{tabular} 
	\end{table}
	
	Given this device the receiver's set of posteriors updated from FB updating is: 
	\begin{equation*}
		\begin{split}
			&p(\omega_{h}|m_{l}) = \{0, 1/4\} \\
			&p(\omega_{h}|m_{l}') = \{1/4, 0\} \\
			&p(\omega_{h}|m_{h}) = \{3/4, 3/4\}
		\end{split}
	\end{equation*}
	where the first and second posterior in each set is updated from $\pi_{1}$ and $\pi_{2}$ respectively. Given these posteriors, the receiver with MEU preference would take action $a_{m}$ when signal $m_{l}$ or $m_{l}'$ is realized and takes action $a_{h}$ when signal $m_{h}$ is realized. The posterior $p(\omega_{h}|m_{l}) = 1/4$ is crucial. Because any posterior assigning a smaller probability to $\omega_{h}$ would induce the receiver to take action $a_{l}$, which makes the sender worse off. 
	
	Given the receiver's action for each signal, the sender's ex-ante payoff from this ambiguous device is given by 
	\begin{equation}\label{ap}
		\left[\frac{1}{2}(1-\lambda) + \frac{1}{3}\lambda\right]u_{s}(a_{m}) + \left[\frac{1}{2}(1-\lambda) + \frac{2}{3}\lambda \right]u_{s}(a_{h})
	\end{equation}
	which is strictly higher than Bayesian persuasion when $\lambda > 0$ and is increasing in $\lambda$. Therefore, the optimal ambiguous persuasion can be approached by letting $\lambda \rightarrow 1$. Notice that $\lambda$ cannot be exactly one. \\
	
	However, the overall probabilities of sending signals given each device also depend on $\lambda$. For example, the overall probability of $m_{l}$ under the two devices are: 
	\begin{align*}
		&\pi_{1}(m_{l}) = \frac{1}{2} \cdot \frac{2}{3}\lambda = \frac{1}{3}\lambda\\
		&\pi_{2}(m_{l}) = \frac{1}{2} \cdot \frac{3}{4}(1-\lambda) + \frac{1}{2} \cdot \frac{1}{4}(1-\lambda) = \frac{1}{2}(1-\lambda)
	\end{align*}
	Clearly, when $\lambda \rightarrow 1$, these two probabilities are severely different, as $\pi_{1}(m_{l})$ goes to $1/3$ and $\pi_{2}(m_{l})$ goes to $0$. Intuitively, knowing $\lambda \rightarrow 1$, whenever the signal $m_{l}$ is observed, the receiver should be almost sure that it is generated by device $\pi_{1}$. However, the crucial posterior $1/4$ inducing action $a_{m}$ is in fact generated by the other device $\pi_{2}$. 
	
	Indeed, for RML updating with any $\alpha > 0$, the receiver would find $a_{l}$ to be strictly better than $a_{m}$. One can also verify that the sender's payoff is equivalent to some non-optimal probabilistic device. Thus, she instead becomes strictly worse off using ambiguous signals compared to using the optimal probabilistic device when the receiver slightly deviates from FB in the direction of ML. 
	
	Nonetheless, one can make the ambiguous device uniform-likelihood by letting $\lambda = 3/5$. In this case, the receiver cannot use likelihood of the realized signal to make any inference about the devices. It is the case where FB and ML coincide, and the receiver will always update with respect to both devices under RML no matter what his $\alpha$ is. As the crucial posterior $1/4$ will always be updated, the sender is able to induce the action $a_{m}$ when signal $m_{l}$ or $m_{l}'$ realizes.  In addition, the sender's payoff will be given by equation (\ref{ap}) with $\lambda = 3/5$, which is strictly higher than Bayesian persuasion. 
	
	Therefore, a uniform-likelihood ambiguous device can guarantee the sender strictly more payoff than Bayesian persuasion in this example. This result suggests that the strict benefit from using ambiguous signals is robust to the concern that the receiver may make likelihood-based inference about the devices from the realized signals. 
	
	\section{Related Literature}\label{discussion}
	This paper adds to the literature on dynamic choice under ambiguity by proposing the Relative Maximum Likelihood updating rule and providing an axiomatic foundation for it. RML is motivated by the idea of using observed information to refine the initial belief in updating. For multiple priors, Full Bayesian \citep*{pires2002rule} and Maximum Likelihood \citep*{GILBOA199333} are the two extremes of refining the initial set of priors according to the likelihood of the observed information. RML is able to capture intermediate behaviors between these two. 
	
	\subsection{Refining the Initial Belief in Updating}
	The idea of refining beliefs as new information arrives is essential in many different non-Bayesian updating rules, and it is certainly not an exclusive feature of updating multiple priors. 
	
	For example, when the initial belief is a single distribution over the states, \cite{Ortoleva2012} characterizes the hypothesis testing updating rule. Under this rule, if the probability of the observed event is too small according to the initial belief, the DM will then find another prior as her revised belief and apply Bayes' rule to this new prior for updating. 
	
	Also in the case when the initial belief is a singleton, \cite{zhao2017} characterizes the Pseudo-Bayesian updating rule when information takes the form ``event $A$ is more likely than event $B$''. In the case where received information contradicts the DM's initial belief, she also finds another prior as her revised belief and applies Bayes' rule to update. 
	
	In cases where the initial belief is a set of priors, one way of refining is discarding priors from the initial set according to some criterion. RML belongs to this category. An updating rule proposed by \cite{Epstein2007} uses likelihood ratio test in statistics as a criterion for discarding priors. More discussions about the differences between RML and their rule can be found in the following subsection. 
	
	The dynamic consistent updating rule characterized in \cite{hanany2007updating} also features discarding priors, and the criterion there is to maintain the optimality of the ex-ante optimal act. 
	
	Yet another way of refining is to consider a different set of priors, where it is possible that some priors are not presented in the initial belief. \cite{Ortoleva2014} characterizes the hypothesis testing updating rule for multiple priors: if the likelihood of the observed event is too low under some prior in the initial belief, then the DM will revise her initial belief and change to another set of priors for updating. 
	
	When the ex-ante preference admit both maxmin expected utility and Choquet expected utility representations, \cite{cheng2021extended} proposes and characterizes the related Extended Relative Maximum Likelihood updating rule. I refer to that paper for a detailed discussion about its relation to RML. 
	
	Beyond the theoretical developments regarding this idea of refining the initial belief, \citet*{DEFILIPPIS2021105188} also identify such behavior in a social learning lab experiment. Their finding suggests that the non-Bayesian behavior observed in the experiment is consistent with an updating rule under which the subjects revise their initial belief according to the information received. In another experiment, \cite{KATHLEENNGANGOUE2021105282} also observes that a non-negligible group of the subjects update their ambiguous belief with ML updating.

	\subsection{Comparison with Likelihood-ratio Updating}\label{likelihood}
	An updating rule proposed by \cite{Epstein2007} and axiomatized by \cite{kovach2015ambiguity} and \cite{HILL2021105209} is also able to capture intermediate behaviors between FB and ML. It  uses a likelihood ratio test as the criterion to discard priors in updating. Formally, given a threshold $\lambda \in [0,1]$, the DM applies Bayes' rule to update a prior $p$ in the set of priors $C$ if and only if the following inequality holds: 
	
	\begin{equation}\label{lrt}
		\frac{p(E)}{\max\limits_{p \in C}p(E)} \geq \lambda
	\end{equation}
	
	Although similar, RML updating is behaviorally different from likelihood-ratio updating (even with a single conditioning event) in the sense that they imply different conditional behaviors. An intuition for that is RML crucially depends on the the set of priors, especially its shape, but the likelihood-ratio updating does not depend on the set of priors in the same way.
	
	In Figure \ref{fig:1}, let the largest triangle represent the simplex of probability distributions over the three states. The set of priors are given by the smaller triangle. Let the conditioning event be $E = \{\omega_{1}, \omega_{2}\}$; then clearly here $C^{*}(E) $ is the extreme point at the bottom of $C$. 
	
	\begin{figure}[h]
		\centering
		\begin{tikzpicture}[x=1.2cm, y = 1.2cm]
			\draw[black] (0,0) -- (5,0) -- (2.5, 4.33) --cycle;
			\node[right] at (5,0) {\small$\omega_{2}$}; 
			\node[right] at (2.5, 4.33) {\small$\omega_{3}$};
			\node[left] at (0,0) {\small$\omega_{1}$};
			\filldraw[blue!20] (2.5,0.4) -- (2,0.7)  -- (2.85,1.5) --cycle;
			\draw[black]  (2.5,0.4) -- (2,0.7)  -- (2.85,1.5) --cycle;
			\draw [black] (2.5,0.4) --(3.2,2.6) -- (1.5,1) -- cycle;
			\node[below] at (2.2,2.2) {\small$C$};
			\node [below] at (2.5, 0.4) {\tiny$C^{*}(\{\omega_{1}, \omega_{2}\})$};
			\node at (3.2, 0.8) {\tiny $C_{\alpha}(\{\omega_{1}, \omega_{2}\})$};
			\draw[thick, dashed,red] (1.5,1.5) -- (3,1.5);
			\draw[dashed, red] (2.5,4.33) -- (3.04, 0);
			\draw[dashed, red] (2.5,4.33) -- (1.2,0);
			
			\draw[thick, dashed, olive] (1.5,0.7) -- (3,0.7);
			\draw[dashed, olive]  (2.5,4.33) -- (1.9,0);
			\draw[dashed, olive]  (2.5,4.33) -- (2.58,0);

			\draw[thick, red] (1.2,0) -- (3.04,0);
			\draw[thick, blue] (1.9,0) -- (3.04,0);
			\draw[thick, olive] (1.9,0) -- (2.58,0);
			
			\draw[decorate,decoration={brace,mirror,raise = 2pt},black] (1.9,0) -- (3.04,0); 
			\node[below] at (2.5,-0.1) {\tiny RML}; 
		\end{tikzpicture}
		\caption{A Graphical Illustration of RML}
		\label{fig:1}
	\end{figure}
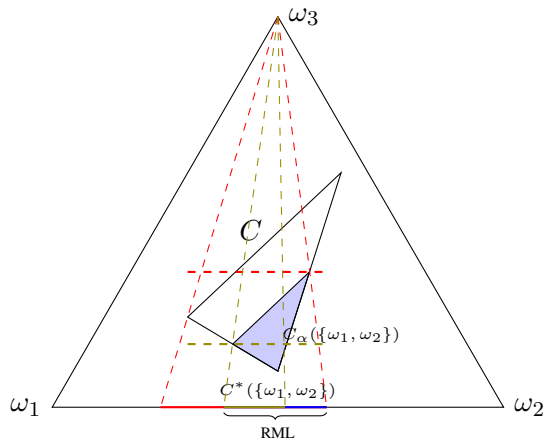
	
	For RML updating with some interior value of  $\alpha $, the set $C_{\alpha}(E)$ is given by the blue shaded triangle. Its projection to the bottom line of the simplex is the set of posteriors. 
	
	However, under the likelihood-ratio updating, the DM updates the subset of $C$ that is below some horizontal line. Two possible horizontal lines are depicted in the graph using dashed lines. Notice that there does not exist a horizontal line that would result in the same subset of priors as $C_{\alpha}(E)$; neither does there exist a horizontal line that would yield the same set of posteriors. Therefore, the conditional preference generated by RML updating here cannot be generated by the likelihood-ratio updating rule. 
	
	More importantly, such a distinction highlights a feature of RML: it does not discard priors based \textit{only} on likelihood. It first uses likelihood to determine the level of discarding (given by the horizontal red dashed line in Figure \ref{fig:1}), and then refines the initial set of priors, taking into account its shape (given by the blue shaded triangle in Figure \ref{fig:1}). This type of refinement may be heuristically motivated by the fact that a DM takes her initial set of priors very seriously, but I would like to highlight that it is fundamentally motivated by the behavioral axioms provided in this paper. 
	
	In fact, the difference between these two updating rules is clearer through the lens of axiomatizations. The conditional preference generated by the likelihood-ratio updating violates one of the key axioms: DC-CS. 
	
	\subsection{Dynamic Choice Under Ambiguity}
	A fundamental issue in dynamic choice under ambiguity is the fact that, for ambiguity sensitive choices, an updating rule cannot preserve both consequentialism and dynamic consistency at the same time \citep*{hanany2007updating,Siniscalchi2009-SINTOO-5}. 
	
	According to the definition in \cite{hanany2007updating}, consequentialism means that conditional preferences should not depend on an event's not occurring or the context of the decision problem (e.g. feasible acts, etc.). Dynamic consistency (which is weaker than the DC axiom in this paper) means that the ex-ante most preferred act should remain optimal to any other acts that are both feasible and agree with it on the event's not occurring after updating. Any updating rule for ambiguous beliefs without other restrictions needs to relax either one of these two properties. 
	
	RML takes the consequentialist approach by requiring consequentialism and relaxing dynamic consistency. The essential implication of consequentialism is that the DM should update her belief in the same way regardless of the decision problem at hand. In other words, the updated belief is given by a function of only the ex-ante belief and the conditioning event. 
	
	As a result, dynamic consistency will sometimes be violated under RML. On the one hand, the violation of dynamic consistency under ambiguity is commonly observed in experiments, e.g. \cite{DOMINIAK2012625} and \cite{GEORGALOS202128}. On the other hand, consistent planning is proposed in the literature as a way to deal with this problem of consequentialist updating rules. With consistent planning, the DM is assumed to be sophisticated such that she is able to anticipate her future decisions and chooses an ex-ante optimal plan accordingly. A behavioral characterization of consistent planning is given by \cite{Siniscalchi2011}. 
	
	Along another route, \cite{hanany2007updating, hanany2009updating} axiomatize updating rules that preserve dynamic consistency yet do not require consequentialism. More specifically, the dynamic consistent updating rules for multiple priors characterized in \cite{hanany2007updating} explicitly depend on the feasible set of acts as well as the optimal acts. Therefore, different decision problems result in different updated beliefs according to this rule. 
	
	Instead of relaxing one of these two properties, \cite{EPSTEIN20031} characterize the rectangularity condition for the set of priors such that dynamic consistency is preserved under FB updating, i.e. consequentialism also holds. Namely, dynamic consistency and consequentialism can both be satisfied when the ex-ante belief takes a specific form. The rectangularity condition, however, imposes a restriction on the possible conditioning events to which FB updating is applicable. 
	
	Another way of keeping both consequentialism and dynamic consistency is to relax reduction of compound evaluations\footnote{Or the law of iterated expectation in \cite{GUL2021105129}'s terminology}, pursued by \cite{Li2020} and \cite{GUL2021105129}. According to this approach, the DM's ex-ante evaluation of an act will also depend on the temporal resolution of uncertainty.
	
	\section{Concluding Remarks}\label{conclude}
	On maxmin (minimax) decision making, in a well-known comment \cite{Good1952} says: ``In what circumstances is a minimax solution reasonable? I suggest that it is reasonable if and only if the least favorable initial distribution is reasonable according to your body of beliefs.'' 
	
	The updating procedure proposed at the beginning of this paper echoes this idea. A prior is going to be updated and thus used to evaluate decisions only if the DM still finds it reasonable upon observing the realized information. 
	
	In this sense, Bayesian updating of a single prior belief is a special case where the initial belief is never revised. In the case where the true probability law governing the uncertainty is known, Bayesian updating is reasonable because one cannot further refine the belief but can only update it conditional on the information received. 
	
	However, in scenarios where the underlying probability law is unknown, the DM needs to form a conjecture about the uncertainty for decision making. Whether the conjecture is a singleton or a set of probabilities, it seems too stringent to require the DM always to stick with her initial conjecture despite new information she might receive. Thus, applying Bayes' rule to the conjecture belief actually reflects confidence about her initial belief. Accordingly, updating rules that do not reflect such confidence and allow for revising the initial belief might also be reasonable. 
	
	Several different updating rules have been proposed in the literature to capture the situation in which initial conjecture is a singleton and may be revised after seeing new information. For when the initial beliefs are multiple priors, this paper proposes the RML updating rule, in which the initial beliefs are revised based on the likelihood of the information observed. More importantly, this paper pinpoints the behaviors that are equivalent to such an updating rule, providing a preference foundation for using likelihood as a criterion in updating. 
	
	This paper also uses RML to address applications involving ambiguity in information design and mechanism design. Many existing results in these areas are derived based on the specific assumption of FB. I have shown RML to be useful for identifying the extent of deviation from FB allowed for  those results to still hold. Last but not least, I illustrate through an example that adopting RML instead of FB as the model of updating helps robustify results from the literature while also maintaining tractability. 
	
	\newpage
	\begin{appendices}
		\counterwithin{table}{section}
		
		\section{Proofs of Results}\label{ap_pf}
		Throughout all the proofs, let $C$ denote the convex and closed set of priors representing the ex-ante preference; let $C^{*}(E)$ denote the subset of $C$ attaining the maximum likelihood of the event $E$: $C^{*}(E) \equiv \{p \in C: p(E) \geq p'(E) ~\forall p' \in C\}$ and let $p^{*}(E)$ denote the maximal probability of event $E$ given $C$: $p^{*}(E) \equiv \max_{p \in C}p(E)$. 
		
		\subsection{Sufficiently Good Consequences}
		Since the characterization results rely largely on the threshold for sufficiently good consequences, I will first derive some useful results here to simplify notations in the proof. First, consider the following lemma. 
		\begin{lemma}\label{lem}
			For all strict $\succsim$-nonnull $E\in \Sigma$, for all $f\in \mathcal{F}$. $\bar{x}_{E,f}$ is a threshold for sufficiently good consequences if and only if 
			\begin{equation*}
				\min\limits_{p \in C} \int_{\Omega} u(f_{E}\bar{x}_{E,f})dp = \min\limits_{p \in C^{*}(E)}  \int_{\Omega} u(f_{E}\bar{x}_{E,f})dp
			\end{equation*}
			Moreover, any $x^{*} \succsim \bar{x}_{E,f}$ is also such a threshold. 
		\end{lemma}
	
		\begin{proof}[Proof of Lemma \ref{lem}]
		\textbf{Sufficiency.} Fix any such $f_{E}\bar{x}_{E,f}$ and find $x$ such that 
		$f_{E}\bar{x}_{E,f} \sim x_{E}\bar{x}_{E,f}$. For any $y \prec z$ and $\lambda \in [0,1]$, by concavity one has, 
		\begin{equation}\label{convex}
		\min\limits_{p \in C} \int_{\Omega} \lambda u( f_{E}\bar{x}_{E,f}) dp  + \min\limits_{p \in C} \int_{\Omega} (1-\lambda) u(y_{E}z)dp \leq	\min\limits_{p \in C} \int_{\Omega} u(\lambda f_{E}\bar{x}_{E,f} + (1-\lambda) y_{E}z)dp
		\end{equation}
		By presumption, the LHS can be obtained by a unique $p \in C^{*}(E)$, therefore the equality holds. Notice that the same also holds for $\lambda x_{E}\bar{x}_{E,f} + (1-\lambda) y_{E}z$, thus the conclusion. 
		
		\textbf{Necessity.} This direction is proved by the contrapositive statement. If 
			\begin{equation*}
		\min\limits_{p \in C} \int_{\Omega} u(f_{E}\bar{x}_{E,f})dp \neq  \min\limits_{p \in C^{*}(E)}  \int_{\Omega} u(f_{E}\bar{x}_{E,f})dp
	\end{equation*}
		then as $u(X)$ is unbounded from above, one can always find $y \prec z$ and $\lambda \in [0,1]$ such that the RHS of \eqref{convex} is obtained at some $p \in C^{*}(E)$, and hence strict inequality holds. 
		\end{proof}
	
		The second statement suggests that for all events $E_{1}, E_{2}$ and acts $f,g$, one can always find the consequence $x^{*} \equiv \max \{\bar{x}_{E_{1},f}, \bar{x}_{E_{1},g}, \bar{x}_{E_{2},f}, \bar{x}_{E_{2},g}\}$ such that is a threshold for any combination of event and act. As a result, all the axioms can be equivalently stated in terms of some $x^{*}$ that is sufficiently good. Therefore, throughout all the remaining proofs, I will use $x^{*}$ to denote such sufficiently good consequence as it should cause no confusion. 
		
		Next, a threshold $\bar{x}_{E,f}$ always exists given Assumption \ref{asum}. 
		\begin{lemma}\label{lem4}
			If $\succsim$ admits MEU representation with $C$ having finitely many extreme points and $u(X)$ being unbounded from above, then for all strict $\succsim$-nonnull $E\in \Sigma$ and $f\in \mathcal{F}$, the threshold $\bar{x}_{E,f} \in X$ with $u(\bar{x}_{E,f})$ finite always exists. 
		\end{lemma}
		
		\begin{proof}[Proof of Lemma \ref{lem4}]
			
			First show that when $C$ contains only finitely many extreme points, an $\bar{x}_{E,f}$ such that $f_{E}\bar{x}_{E,f}$ is evaluated at some extreme point in $C^{*}(E)$ always exists. 
			
			Let $q$ be any extreme point of $C^{*}(E)$ and let $p$ be any extreme point of $C$. The act $f_{E}\bar{x}_{E,f}$ is evaluated at $q$ if for all extreme points $p$ of $C$, 
			\begin{equation}\label{eq1}
				\int_{E} u(f) dq + (1-q(E)) u(\bar{x}_{E,f}) \leq \int_{E} u(f) dp + (1-p(E)) u(\bar{x}_{E,f})
			\end{equation}
			Notice that the first term of both LHS and RHS does not depend on $\bar{x}_{E,f}$, furthermore, $(1-q(E)) \leq (1-p(E))$ as $q \in C^{*}(E)$. When $p$ is also in $C^{*}(E)$, the value of $\bar{x}_{E,f}$ does not matter and one can pin down the $q \in C^{*}(E)$ such that minimizes the evaluation of $f_{E}\bar{x}_{E,f}$ among all extreme points in $C^{*}(E)$. Denote the minimizing $q$ by $q^{*}$. 
			
			Fix $q^{*}$, for an extreme point $p$ not in $C^{*}(E)$, the following inequality becomes strict: $(1-q^{*}(E)) < (1-p(E))$. Then for all $E,f$ and such an extreme point $p$, inequality (\ref{eq1}) holds if 
			\begin{equation*}
				u(\bar{x}_{E,f,p}) \geq \frac{\int_{E} u(f) dq^{*} -  \int_{E} u(f) dp}{q^*{(E) - p(E)}}
			\end{equation*}
			Since $u(X)$ is unbounded from above, such an $\bar{x}_{E,f,p}$ always exist. Then it suffices to let $\bar{x}_{E,f} = \max \{ \max_{p} \bar{x}_{E,f, p}, x\}$  for $x \sim_{E} f$ where both ``max'' are according to the ex-ante preference. Since there are only finitely many extreme points $p$ in $C$, the maximum over a finite set always exists. 
		\end{proof}		
		
		\subsection{Proof of Theorem \ref{thm1}}
		\textbf{Necessity.} Suppose the conditional preferences are given by ML updating. For any strict $\succsim$-nonnull $E\in \Sigma$, fix any $f\in \mathcal{F}$ and $x\in X$ such that $f\sim_{E}x$. Then 
		\begin{equation*}
		\min\limits_{p \in C^{*}(E)} \int_{E} u(f)\frac{dp}{p(E)}= u(x)
		\end{equation*}
		and one can further derive
		\begin{align*}
			\min\limits_{p \in C} \int_{\Omega} u(f_{E}x^{*}) dp &= \min\limits_{p \in C^{*}(E)} \int_{\Omega} u(f_{E}x^{*}) dp \\
			&= \min\limits_{p\in C^{*}(E)}  \left[ \int_{E}  u(f_{E}x^{*}) \frac{dp}{p(E)} \cdot p(E) + (1-p(E))u(x^{*}) \right] \\
			& = p^{*}(E)\cdot \min\limits_{p\in C^{*}(E)} \int_{E} u(f)\frac{dp}{p^{*}(E)} +  (1-p^{*}(E))u(x^{*}) \\
			&= p^{*}(E) u(x) + (1-p^{*}(E))u(x^{*}) = \min\limits_{p \in C} \int_{\Omega}u(x_{E}x^{*})dp
		\end{align*}
		where the third equality follows from $p(E) = p^{*}(E)$ for all $p \in C^{*}(E)$, the last equality follows from the fact that $x^{*} \succsim x$. 
		
		\textbf{Sufficiency.} For the sufficiency of CR-S, fix any strict $\succsim$-nonnull $E\in \Sigma$, consider the contra positive statement: not ML updating implies not CR-S. 
		
		Let $C_{E}$ be the closed and convex set of posteriors representing the conditional preference $\succsim_{E}$. Not ML updating implies that $C_{E} \neq \{p/p(E): p \in C^{*}(E)\}$. In other words, either there exists $\tilde{p} \in C^{*}(E)$ such that $\tilde{p}/\tilde{p}(E) \notin C_{E}$, or there exists $q \in C_{E}$ such that $q \notin \{p/p(E): p \in C^{*}(E)\}$ or both. 
		
		For the two different cases of not ML, since both $C_{E}$ and $\{p/p(E): p \in C^{*}(E)\}$ are convex and closed set, the same type of separating hyperplane argument can be applied to both cases. Thus the proof here only shows the implication of the first case, while the same argument applies to the other one. 
		
		Formally, in the first case, there exists $\tilde{p} \in C^{*}(E)$ such that $\tilde{p}/\tilde{p}(E) \notin C_{E}$. As $\{\tilde{p}/\tilde{p}(E)\}$ is compact and $C_{E}$ is convex and closed, the strict separating hyperplane theorem implies that there exists an act $f\in \mathcal{F}$ such that
		\begin{equation*}
			\min\limits_{p\in C^{*}(E)} \int_{E} u(f) \frac{dp}{p(E)} \leq  \int_{E} u(f) \frac{d\tilde{p}}{\tilde{p}(E)} < \min\limits_{p\in C_{E}} \int_{E} u(f)dp = u(x)
		\end{equation*}
		
		Then it further implies
		\begin{equation*}
			\begin{split}
				\min\limits_{p \in C} \int_{\Omega} u(f_{E}x^{*}) dp & = \min\limits_{p \in C^{*}(E)}  \left[ \int_{E} u(f) \frac{dp}{p(E)} \cdot p(E) + u(x^{*})(1-p(E)) \right]\\
				& =p^{*}(E) \min\limits_{p \in C^{*}(E)} \int_{E} u(f) \frac{dp}{p(E)} +  u(x^{*})(1-p^{*}(E))\\
				& < p^{*}(E) \cdot u(x) + u(x^{*})(1-p^{*}(E)) = \min\limits_{p \in C} \int_{\Omega} u(x_{E}x^{*}) dp
			\end{split}
		\end{equation*}
		Therefore, CR-S is not true. The proof of the other case is analogously the same.\qed
		
		\subsection{Proof of Theorem \ref{thm2}}
		
		The necessity of CR-UO is given in the main text. For DC-CS, reversing the arguments in step 4 of proving sufficiency shows that equation (\ref{equ3}) is necessary under Contingent RML. It then immediately implies that DC-CS needs to be true. 
		
		For the sufficiency, fix any strict $\succsim$-nonnull $E\in \Sigma$, the proof proceeds by the following steps: 
		
		\underline{Step 1.} When CR-UO is true, the following inequalities holds: $f_{E}x^{*} \succsim x_{E}x^{*} \succsim x \succsim f_{E}x$, which further implies that there always exists an $\alpha[E,f] \in [0,1]$ such that: (see Figure \ref{fig1}) 
		\begin{equation}\label{equ1}
			\alpha[E,f] U(f_{E}x^{*}) + (1-\alpha[E,f]) U(f_{E}x) = \alpha[E,f] U(x_{E}x^{*}) +  (1-\alpha[E,f]) U(x) 
		\end{equation}
		$\alpha[E,f]$ is independent of the value of $x^{*}$ as $f_{E}x^{*}$ is always evaluated at extreme points in $C^{*}(E)$. Then both $U(f_{E}x^{*})$ and $U(x_{E}x^{*})$ have the common term $u(x^{*})(1- p^{*}(E))$ which cancels out. Furthermore, it is easy to see that this $\alpha[E,f]$ is unique if either $f_{E}x \prec x$ or $f_{E}x^{*} \succ x_{E}x^{*}$ hold.
		
		\underline{Step 2. } Construction of acts satisfying the premises of DC-CS. 
		
		\textbf{Equivalence Class.} First, for any two acts $f, f'$, denote them by $f \equiv_{E} f'$ if $f(\omega) = f'(\omega)$ for all $\omega \in E$. Then an equivalence class of acts can be accordingly defined: 
		\begin{equation*}
			[f] = \{ f' \in \mathcal{F}: f' \equiv_{E} f \}
		\end{equation*}
		By step 1, $\alpha [E, f] = \alpha[E, f']$ whenever $f' \in [f]$. Hereafter, I use $f$ to denote the whole class of acts $[f]$, as it should cause no confusion.  \\
		
		\textbf{When $\alpha[E]$ is not unique.} If there does not exist any $f\in \mathcal{F}$ with $\alpha[E, f]$ being unique. Then by step 1, for all $f\in \mathcal{F}$ and $x\in X$ with $f \sim_{E}x$, it implies that both $f_{E}x \sim x$ and $f_{E}x^{*} \sim x_{E}x^{*}$ hold. 
		
		Then it is the case in which both CR-C and CR-S hold at the same time. Namely, the conditional preference $\succsim_{E}$ can be represented by both FB and ML. As $E$ is strict $\succsim$-nonnull, it implies that $C = C^{*}(E)$. In other words, $C \neq C^{*}(E)$ implies there exists at least an $f\in \mathcal{F}$ such that $\alpha[E,f]$ is unique. \\
		
		When unique $\alpha[E,f]$ exists, fix some $f\in \mathcal{F}$ such that $\alpha[E, f]$ is unique.
		
		\textbf{$\alpha[E,f]$ is unchanged under mixture with constant acts.} For any $y \in X$, for any $\lambda \in (0,1]$ consider the act $f_{\lambda}y = \lambda f + (1-\lambda)y $, which is an Anscombe-Aumann mixture of acts.\footnote{The case $\lambda = 0$ is excluded since when $\lambda = 0$, $f_{\lambda}y$ coincides with $y$, and then $\alpha[E, y]$ is not unique.} By certainty independence, $f\sim_{E}x$ implies that $f_{\lambda} y \sim_{E}  \lambda x + (1-\lambda)y$. Let $x_{\lambda}y\in X$ denote the consequence indifferent to $\lambda x + (1-\lambda)y$. Then for all $x^{*}$ that is sufficiently good for $f$ and $f_{\lambda}y$, equation (\ref{equ1}) for $f_{\lambda}y$ implies that 
		\begin{equation}\label{equ6}
			\alpha[E, f_{\lambda}y] U([f_{\lambda}y]_{E}x^{*}) + (1-\alpha[E, f_{\lambda}y]) U([f_{\lambda}y]_{E}x_{\lambda}y) =  \alpha[E, f_{\lambda}y] U([x_{\lambda}y]_{E}x^{*}) + (1-\alpha[E, f_{\lambda}y]) U(x_{\lambda}y) 
		\end{equation}
		Moreover, $\alpha[E, f_{\lambda}y] $ is unique. Notice that 
		\begin{align*}
			U([f_{\lambda}y]_{E}x_{\lambda}y) = U([f_{E}x]_{\lambda}y)=  \lambda U(f_{E}x) + (1-\lambda)U(y)
		\end{align*}
		where the first equality follows from Anscombe-Aumann mixture, the second equality follows from certainty independence. Then one can  write equation (\ref{equ6}) as
		\begin{align*}
			&\alpha[E, f_{\lambda}y] \min\limits_{p\in C^{*}(E)} \int_{E} u(f_{\lambda}y)dp  + (1-\alpha[E, f_{\lambda}y]) [\lambda U(f_{E}x) + (1-\lambda)U(y)] \\
			&=\alpha[E, f_{\lambda}y] u(x_{\lambda}y) p^{*}(E) +  (1-\alpha[E, f_{\lambda}y]) [\lambda U(x) + (1-\lambda)U(y)] 
		\end{align*}
		which further implies that 
		\begin{equation*}
			\alpha[E, f_{\lambda}y] U(f_{E}x^{*}) + (1-\alpha[E, f_{\lambda}y]) U(f_{E}x)  = \alpha[E, f_{\lambda}y] U(x_{E}x^{*}) + (1-\alpha[E, f_{\lambda}y]) U(x)
		\end{equation*}
		As $f\sim_{E}x$ one also has
		\begin{equation*}
			\alpha[E, f] U(f_{E}x^{*}) + (1-\alpha[E, f]) U(f_{E}x)= \alpha[E, f] U(x_{E}x^{*}) + (1-\alpha[E, f]) U(x) 
		\end{equation*}
		Since $\alpha[E,f]$ is unique, it has to be the case that $\alpha[E,f_{\lambda}y] = \alpha[E,f]$ for all $\lambda\in (0,1]$ and $y\in X$. \\
		
		\textbf{Construction of acts satisfying premises of DC-CS.} Recall that the utility function $u(\cdot)$ is normalized such that $u(X) = (\underline{u}, \infty)$ for some $\underline{u} \in \R_{<0}\cup \{-\infty\}$. 
		
		For the fixed $f\in \mathcal{F}$, for any $\epsilon > 0$, there always exists an act $f_{\epsilon}$ given by taking mixtures between $f$ and some constant act such that $u(f_{\epsilon}(\omega)) \in [0, \epsilon]$ for all $\omega \in \Omega$. By previous argument, one has $\alpha[E,f_{\epsilon}] = \alpha[E,f]$. 
		
		Take any act $g\in \mathcal{F}$ that cannot be obtained from $f$ by taking mixtures with constant acts. (If can, then $\alpha[E,g] = \alpha[E,f]$) For any $\lambda \in (0,1]$ and $y \in X$, consider the act $g_{\lambda}y$, i.e. the $\lambda$ mixture between $g$ and $y$. 
		
		For the two premises of DC-CS\footnote{Notice here $f_{\epsilon}$ takes the role of $f$ and $g_{\lambda}y$ takes the role of $g$ in the statement of this axiom.}, it suffices to find $\lambda \in (0,1]$ and $y\in X$ such that the following two equations hold: 
		\begin{equation}\label{dc1}
			U(f_{\epsilon E}x) = U([g_{\lambda}y]_{E}x) 
		\end{equation}
		for $x \sim_{E} f_{\epsilon}$  and 
		\begin{equation}\label{dc2}
			U(f_{\epsilon E}x^{*}) = U([g_{\lambda}y]_{E}x^{*}) 
		\end{equation}
		for some $x^{*}$ that is sufficiently good for $f_{\epsilon}$ and $g_{\lambda}y$. The sufficiently good $x^{*}$ can be chosen after $f_{\epsilon}$ and $g_{\lambda}y$ are pinned down such that both $U(f_{\epsilon E}x^{*})$ and $U([g_{\lambda}y]_{E}x^{*})$ will both be evaluated at extreme points in $C^{*}(E)$. 
		
		Consider the extreme points of the set of posteriors of $C^{*}(E)$, let $q_{f}$ and $q_{g}$ denote the two of them evaluating $f_{\epsilon}$ and $g_{\lambda}y$ respectively, then from equation (\ref{dc2}) one can derive
		\begin{align*}
			U(f_{\epsilon E}x^{*}) & = U([g_{\lambda}y]_{E}x^{*})  \\
			\min\limits_{p \in C^{*}(E)} \int_{\Omega} u(f_{\epsilon E}x^{*}) dp & = \min\limits_{p \in C^{*}(E)} \int_{\Omega} u([g_{\lambda}y]_{E}x^{*}) dp  \\
			\min\limits_{p \in C^{*}(E)} \int_{E} u(f) dp + u(x^{*})(1-p^{*}(E)) & =  \min\limits_{p \in C^{*}(E)} \int_{E} u(g_{\lambda}y) dp + u(x^{*})(1-p^{*}(E)) \\
			\min\limits_{p \in C^{*}(E)} \int_{E} u(f) dp & =  \min\limits_{p \in C^{*}(E)} \int_{E} u(g_{\lambda}y) dp  \\
			\min\limits_{p \in C^{*}(E)} \int_{E} u(f) \frac{dp}{p^{*}(E)} &= \min\limits_{p \in C^{*}(E)} \int_{E} u(g_{\lambda}y) \frac{dp}{p^{*}(E)}
		\end{align*}
		i.e. 
		\begin{equation*}
			u(f_{\epsilon})\cdot q_{f}= \lambda u(g)\cdot q_{g} + (1-\lambda)u(y)
		\end{equation*}
		Thus, for each $\lambda \in(0,1]$, the constant act $y$ could be pinned down by letting 
		\begin{equation}\label{equ2}
			(1-\lambda)u(y) = u(f_{\epsilon})\cdot q_{f} - \lambda u(g)\cdot q_{g}
		\end{equation}
		and $y$ is arbitrary if $\lambda = 1$. 
		
		Next consider equation (\ref{dc1}): 
		\begin{equation*}
			\min\limits_{p\in C} \int_{\Omega} u(f_{\epsilon E}x) dp = \min\limits_{p\in C} \left\{ \lambda \int_{E} u(g) dp + (1-\lambda) u(y)p(E)  + (1-p(E))u(x) \right\}
		\end{equation*}
		
		Consider the RHS of this equation and plug into $(1-\lambda)u(y)$ from equation (\ref{equ2}) yields
		\begin{align*}
			&\quad \min\limits_{p\in C} \left\{ \lambda \int_{E}  u(g) dp + [u(f_{\epsilon})\cdot q_{f} - \lambda u(g)\cdot q_{g}]p(E)  + (1-p(E))u(x) \right\} \\
			& = \min\limits_{p\in C} \left\{u(f_{\epsilon})\cdot q_{f} p(E) + (1-p(E))u(x)  + \lambda \left[\int_{E}  u(g) dp  - u(g)\cdot q_{g}p(E) \right]  \right\} 
		\end{align*}
		i.e. 
		\begin{equation}\label{equ7}
			\min\limits_{p\in C} \int_{\Omega} u(f_{\epsilon E}x) dp = \min\limits_{p\in C} \left\{u(f_{\epsilon})\cdot q_{f} p(E) + (1-p(E))u(x)  + \lambda \left[\int_{E}  u(g) dp  - u(g)\cdot q_{g}p(E) \right]  \right\} 
		\end{equation}
		Therefore, equation (\ref{dc1}) holds if there exists $\lambda \in (0,1]$ that solves equation (\ref{equ7}). Notice that the LHS of equation (\ref{equ7}) does not depend on $\lambda$, meanwhile the RHS is a continuous function of  $\lambda$, which I further denoted it by $R(\lambda)$. 
		
		Given continuity, the existence of a solution to equation (\ref{equ7}) can be proved by showing that $R(0) >$ LHS and $R(1) <$ LHS. 
		
		First consider $R(0)$: 
		\begin{align*}
			R(0) & = \min\limits_{p\in C} \left\{u(f_{\epsilon})\cdot q_{f} p(E) + (1-p(E))u(x)  \right\} \\
			& \geq  \min\limits_{p\in C} \left\{ u(x) p(E) + (1-p(E))u(x)\right\} \\
			& = u(x) \\
			& > \min\limits_{p\in C} \int_{\Omega} u(f_{\epsilon E}x) dp = \text{LHS} 
		\end{align*}
		where the second and forth inequality comes from $f_{\epsilon E}x^{*} \succsim x_{E}x^{*}$ and $x \succsim f_{\epsilon E}x$ respectively. The forth inequality is strict because of the fact that $\alpha[E,f]$ is unique, thus one of the inequalities has to be strict.
		
		Next consider $R(1)$: 
		\begin{align*}
			R(1) & = \min\limits_{p\in C} \left\{u(f_{\epsilon})\cdot q_{f} p(E) + (1-p(E))u(x)  + \int_{E}  u(g) dp  - u(g)\cdot q_{g}p(E)   \right\}\\
			&  \leq \epsilon + \min\limits_{p\in C} \left\{ \int_{E} u(g) dp  - u(g)\cdot q_{g}p(E)  \right\} \\
			& =  \epsilon +\min\limits_{p\in C} \left\{ p(E) \cdot \left[ \int_{E} u(g) \frac{dp}{p(E)}  - u(g)\cdot q_{g}\right]\right\} 
		\end{align*}
		where the inequality follows from $u(f_{\epsilon}) \leq \epsilon$.
		
		For the second term, its minimum is $0$ if it is the case $\min\limits_{p\in C} \int_{E} u(g) \frac{dp}{p(E)}  - u(g)\cdot q_{g} = 0$. In this case, notice that conditional evaluation of $g$ coincides under FB and ML updating. That is, both CR-C and CR-S holds for $g$, then $\alpha[E,g]$ is not unique. It suffices to let $\alpha[E,g] = \alpha[E,f]$. 
		
		On the other hand, if it is the case $\min\limits_{p\in C} \int_{E} u(g) \frac{dp}{p(E)}  - u(g)\cdot q_{g} < 0$, the minimum of the second term is negative. Then it suffices to find $\epsilon > 0$ such that 
		\begin{equation*}
			\epsilon < - \min\limits_{p\in C} \left\{ p(E) \cdot \left[ \int_{E} u(g) \frac{dp}{p(E)}  - u(g)\cdot q_{g}\right]\right\} 
		\end{equation*}
		Given this $\epsilon$, it further implies that
		\begin{equation*}
			R(1) < \epsilon  - \epsilon = 0 \leq \min\limits_{p\in C} \int_{\Omega} u(f_{\epsilon E}x) dp = \text{LHS} 
		\end{equation*}
		Therefore, the existence of $\lambda \in (0,1)$ that solves equation (\ref{equ7}) is guaranteed. 
		
		Finally, once $\lambda$ is solved, $u(y)$ is given by equation (\ref{equ2}): 
		\begin{equation*}
			u(y) = \frac{u(f_{\epsilon})\cdot q_{f} - \lambda u(g)\cdot q_{g}}{1-\lambda} \geq - \frac{\lambda}{1-\lambda} u(g)\cdot q_{g}
		\end{equation*}
		where the last inequality comes from $u(f_{\epsilon}(\omega)) \geq 0$. It remains to show that $u(y) > \underline{u}$ to guarantee the existence of this construction. It suffices to transform $g$  by taking mixtures with constant acts before the construction to get $u(g)\cdot q_{g} \leq 0$, then it would imply $u(y) \geq 0 > \underline{u}$ as desired.

		\underline{Step 3.} DC-CS implies $\alpha[E, f]$ to be a constant across all $f\in \mathcal{F}$, and it is unique if $C \neq C^{*}(E)$.
		
		In the following, I abuse notation to use $f$ to denote $f_{\epsilon}$ and $g$ to denote $g_{\lambda}y$ be the pair of acts constructed in the last step satisfying the two premises of DC-CS. 
		
		By step 1, $f\sim_{E} x$ implies that 
		\begin{equation}\label{equ11}
			\alpha[E,f] U(f_{E}x^{*}) + (1-\alpha[E,f])U(f_{E}x) =  \alpha[E,f] U(x_{E}x^{*}) + (1-\alpha[E,f]) U(x) 
		\end{equation}
		DC-CS implies $g \sim_{E} x$ and thus, 
		\begin{equation}\label{equ10}
			\alpha[E,g] U(g_{E}x^{*}) + (1-\alpha[E,g]) U(g_{E}x) = \alpha[E,g] U(x_{E}x^{*}) + (1-\alpha[E,g]) U(x) 
		\end{equation}
		Consider the LHS of equation (\ref{equ10}) and denote it by $L$: 
		\begin{align*}
			L & = \alpha[E, g] U(g_{E}x^{*}) + (1-\alpha[E, g])  U(g_{E}x)  \\
			& = \alpha[E, f] U(g_{E}x^{*}) + (1-\alpha[E, f]) U(g_{E}x)+ [\alpha[E, g] - \alpha[E, f]](U(g_{E}x^{*}) - U(g_{E}x))\\
			& \equiv L' +[\alpha[E, g] - \alpha[E, f]] M_{1} 
		\end{align*}
		Meanwhile the RHS of equation (\ref{equ10}) denoted by $R$ can be further derived as
		\begin{align*}
			R & = \alpha[E, g] U(x_{E}x^{*})  +  (1-\alpha[E, g]) U(x)\\
			& = \alpha[E, f] U(x_{E}x^{*})+ (1-\alpha[E, f]) U(x) + [\alpha[E, g] - \alpha[E, f]](U(x_{E}x^{*}) - U(x))\\
			& \equiv R' + [\alpha[E, g] - \alpha[E, f]] M_{2} 
		\end{align*}
		Notice that by equation (\ref{equ11}), $R'$ also equals to $(1-\alpha[E, f])U(f_{E}x) + \alpha[E, f] U(f_{E}x^{*})$. Thus $L' = R'$ as $f_{E}x \sim g_{E}x$ and $f_{E}x^{*} \sim g_{E}x^{*}$ hold. 
		
		Then the fact $L = R$ implies 
		\begin{align*}
			&L - R  =  L' +[\alpha[E, g] - \alpha[E, f]] M_{1} - R' - [\alpha[E, g] - \alpha[E, f]] M_{2} = 0 \\ 
			&\Rightarrow L' - R' =  [\alpha[E, f] - \alpha[E, g]] [M_{1} - M_{2}]
		\end{align*}
		Further notice that 
		\begin{align*}
			M_{1} - M_{2}& =  [U(g_{E}x^{*}) - U(g_{E}x)] - [ U(x_{E}x^{*}) - U(x)] \\
			& = \left[\min\limits_{p\in C} \int_{\Omega}u(g_{E}x^{*})dp - \min\limits_{p\in C} \int_{\Omega}u(g_{E}x)dp\right] - [u(x)p^{*}(E) + u(x^{*})(1-p^{*}(E)) - u(x)] \\
			& = \left[\min\limits_{p\in C^{*}(E)} \int_{\Omega}u(g_{E}x^{*})dp - \min\limits_{p\in C} \int_{\Omega}u(g_{E}x)dp\right] - [u(x^{*}) - u(x)](1-p^{*}(E)) \\
			& = \left[\min\limits_{p\in C^{*}(E)} \int_{\Omega}u(g_{E}x^{*})dp - \min\limits_{p\in C^{*}(E)} \int_{\Omega}u(g_{E}x)dp\right] - [u(x^{*}) - u(x)](1-p^{*}(E)) \\
			& + \left[\min\limits_{p\in C^{*}(E)} \int_{\Omega}u(g_{E}x)dp - \min\limits_{p\in C} \int_{\Omega}u(g_{E}x) dp\right] \\
			& = \min\limits_{p\in C^{*}(E)} \int_{\Omega}u(g_{E}x)dp - \min\limits_{p\in C} \int_{\Omega}u(g_{E}x)dp
		\end{align*}
		where the last equality follows from 
		\begin{align*}
			&\quad\min\limits_{p\in C^{*}(E)} \int_{\Omega}u(g_{E}x^{*})dp - \min\limits_{p\in C^{*}(E)} \int_{\Omega}u(g_{E}x) dp\\
			&= \left[\min\limits_{p\in C^{*}(E)} \int_{E}u(g)\frac{dp}{p(E)} - \min\limits_{p\in C^{*}(E)} \int_{E}u(g)\frac{dp}{p(E)}\right]p^{*}(E) + [u(x^{*}) - u(x)](1-p^{*}(E)) 
		\end{align*}
		
		Next, consider the following lemma: 
		\begin{lemma}\label{lem5}
			For any $f\in \mathcal{F}$ such that $\alpha[E, f]$ is unique, if $f\sim_{E}x$ then 
			\begin{equation*}
				\min\limits_{p\in C} \int_{\Omega} u(f_{E}x) dp < \min\limits_{p\in C^{*}(E)} \int_{\Omega} u(f_{E}x) dp
			\end{equation*}
		\end{lemma}
		
		\begin{proof}[Proof of Lemma \ref{lem5}]
			Since if $\alpha[E, f]$ is unique, it implies that for $f\sim_{E}x$, one of the inequalities $f_{E}x^{*} \succsim x_{E}x^{*}$ and $x \succsim f_{E}x$ is strict. Notice that the first inequality implies
			\begin{equation*}
				\min\limits_{p\in C^{*}(E)} \int_{\Omega} u(f)dp \geq u(x)p^{*}(E)
			\end{equation*}
			and the second implies that 
			\begin{equation*}
				u(x) \geq \min\limits_{p\in C} \int_{\Omega} u(f_{E}x)dp
			\end{equation*}
			
			Add the term $u(x)(1-p^{*}(E))$ to both sides of the first inequality, then combine both inequalities and recall that one of them has to be strict yield:
			\begin{equation*}
				\min\limits_{p\in C^{*}(E)} \int_{\Omega} u(f)dp + u(x)(1-p^{*}(E)) >  \min\limits_{p\in C} \int_{\Omega} u(f_{E}x)dp
			\end{equation*}
			which is equivalent to $\min\limits_{p\in C^{*}(E)} \int_{\Omega} u(f_{E}x) dp > \min\limits_{p\in C} \int_{\Omega} u(f_{E}x) dp $. 
		\end{proof}
		Define
		\begin{equation*}
			\Delta f_{E}x \equiv  \min\limits_{p\in C^{*}(E)} \int_{\Omega} u(f_{E}x)dp - \min\limits_{p\in C} \int_{\Omega} u(f_{E}x)dp 
		\end{equation*}
		
		If $\alpha[E,g]$ is not unique, then it suffice to let $\alpha[E,g] = \alpha[E,f]$. If it is also unique, given Lemma \ref{lem5}, $g\sim_{E}x$ implies that 
		\begin{equation*}
			\Delta g_{E}x > 0
		\end{equation*}
		
		Therefore the difference
		\begin{equation*}
			L' - R' =  [\alpha[E, f] - \alpha[E, g]] [M_{1} - M_{2}] = [\alpha[E, f] - \alpha[E, g]]\cdot \Delta g_{E}x
		\end{equation*}
		is 0 if and only if $\alpha[E, f] = \alpha[E, g]$. 
		
		Notice that the construction in step 2 can be applied to any $g \in \mathcal{F}$. It then implies that $\alpha[E,f]$ needs to be a constant across all $f\in \mathcal{F}$. Therefore, equation (\ref{equ1}) now can be written as 
		\begin{equation}\label{equ3}
			\alpha[E] U(f_{E}x^{*}) + (1-\alpha[E]) U(f_{E}x) = \alpha[E] U(x_{E}x^{*}) + (1-\alpha[E]) U(x) 
		\end{equation}
		
		\underline{Step 4.} Equation (\ref{equ3}) implies that the DM's conditional evaluation of any $f\in \mathcal{F}$ can be represented by
		\begin{equation*}
			\min\limits_{p\in C_{\alpha[E]}(E)} \int_{E} u(f) \frac{dp}{p(E)} = u(x)
		\end{equation*}
		i.e. is given by RML updating with $\alpha[E]$. 
		\bigskip
		
		For any sufficiently good $x^{*}$, the LHS of equation (\ref{equ3}) can further be derived as: 
		\begin{align*}
			&\quad \alpha[E] U(f_{E}x^{*}) + (1-\alpha[E]) U(f_{E}x) \\
			& = \alpha[E] \min\limits_{q\in C^{*}(E)} \int_{\Omega} u(f_{E}x^{*})dq + (1-\alpha[E])  \min\limits_{p\in C} \int_{\Omega} u(f_{E}x) dp  \\
			& = \alpha[E]\left[ \min\limits_{q\in C^{*}(E)}  \int_{E} u(f) \frac{dq}{p^{*}(E)}\cdot p^{*}(E)  +(1-p^{*}(E)) u(x^{*})  \right] \\
			& \quad + (1-\alpha[E]) \min\limits_{p\in C} \int_{\Omega} u(f_{E}x) dp\\
			& = \alpha[E]\left[ \min\limits_{q\in C^{*}(E)}  \int_{E} u(f) \frac{dq}{p^{*}(E)}\cdot p^{*}(E)  + (1-p^{*}(E)) u(x)  \right]   \\ 
			&\quad+ \alpha[E] [u(x^{*}) - u(x)][1-p^{*}(E)] + (1-\alpha[E]) \min\limits_{p\in C} \int_{\Omega} u(f_{E}x) dp
		\end{align*}
		
		On the other hand, the RHS of equation (\ref{equ3}) can also be derived as
		\begin{align*}
			\alpha[E] U(x_{E}x^{*}) +  (1-\alpha[E]) U(x) =  \alpha[E] [u(x^{*}) - u(x)][1-p^{*}(E)] + u(x) 
		\end{align*}
		Observe that now equalizing the LHS and RHS of equation (\ref{equ3}) implies
		\begin{align*}
			u(x) & =   \alpha[E]\left[ \min\limits_{q\in C^{*}(E)}  \int_{E} u(f) \frac{dq}{p^{*}(E)}\cdot p^{*}(E)  +(1-p^{*}(E)) u(x)  \right] + (1-\alpha[E]) \min\limits_{p\in C} \int_{\Omega} u(f_{E}x) dp \\
			& =  \alpha[E]  \min\limits_{q\in C^{*}(E)}  u(f_{E}x)dq + (1-\alpha[E]) \min\limits_{p\in C} \int_{\Omega} u(f_{E}x) dp \\
			& = \min\limits_{q\in C^{*}(E)} \min\limits_{p\in C}  \int_{\Omega} u(f_{E}x) d((\alpha[E])q + (1-\alpha[E]) p )\\
			& = \min\limits_{p \in C_{\alpha[E]}(E)} \int_{\Omega} u(f_{E}x)dp
		\end{align*}
		where $C_{\alpha[E]}(E) = \alpha[E]C^{*}(E) + (1-\alpha[E])C $. 
		
		From the last equality one can further derive
		\begin{align*}
			0  & = \min\limits_{p \in C_{\alpha[E]}(E)} \int_{\Omega} u(f_{E}x)dp - u(x) \\
			& =  \min\limits_{p \in C_{\alpha[E]}(E)} \int_{\Omega} [u(f_{E}x) - u(x)]dp \\
			& = \min\limits_{p \in C_{\alpha[E]}(E)}  \int_{E} [u(f) - u(x)]dp \\
			& = \min\limits_{p \in C_{\alpha[E]}(E)} \left[ \int_{E} u(f)dp - u(x)p(E) \right]
		\end{align*}
		
		When $E$ is strict $\succsim$-nonnull, $p(E) > 0$ for all $p \in C_{\alpha[E]}(E)$, then the last equality further implies 
		\begin{equation*}
			\min\limits_{p\in C_{\alpha[E]}(E)} \int_{E} u(f) \frac{dp}{p(E)} = u(x)
		\end{equation*}
		which represents the conditional evaluation of $f$ under $\succsim_{E}$ since $f \sim_{E}x$. \\
		
		Therefore, the conditional preference $\succsim_{E}$ for each strict $\succsim$-nonnull event $E$ is given by RML updating with $\alpha[E]$. 
		
		\qed

		\subsection{Proof of Theorem \ref{thm4}}
		Given Theorem \ref{thm2}, the only remaining proof here is to show that $\alpha[E]$ is a constant across all events if and only if the EC axiom holds. 
		
		The necessity of the EC axiom is also immediate when one plugs a constant $\alpha$ into equation (\ref{equ3}). \\
		
		In the following I show that EC implies $\alpha[E]$ to be a constant across all strict $\succsim$-nonnull events.
		
		First of all, consider the case that there does not exist any strict $\succsim$-nonnull $E\in \Sigma$ such that $C \neq C^{*}(E)$. Then it implies that all $p \in C$ agree with the probability of all strict $\succsim$-nonnull event $E$. 
		
		Thus, if there exists only one strict $\succsim$-nonnull event $E$ such that $C \neq C^{*}(E)$, i.e. $\alpha[E]$ is unique, then it suffices to let this $\alpha[E]$ to be the constant $\alpha$ across all events.
		
		Next, when there exists at least two strict $\succsim$-nonnull events, $E_{1}$ and $E_{2}$, such that both $\alpha[E_{1}]$ and $\alpha[E_{2}]$ are unique. 
		
		Similar to the construction in the proof of Theorem \ref{thm2}, fix any $f\in \mathcal{F}$.  For any $\epsilon > 0$, let $f_{\epsilon}$ denote the act given by taking mixtures between $f$ and some constant act such that $u(f_{\epsilon}(\omega)) \in [0, \epsilon]$ for all $\omega \in \Omega$. 
		
		Take any $g\in \mathcal{F}$, let $g_{\lambda}y$ denote the act given by taking mixtures between $g$ and $y\in X$ with $\lambda \in (0,1]$. 
		
		Then for the two premises of EC, it suffices to find $\lambda$ and $y$ such that the following equations hold: 
		\begin{equation}\label{ec1}
			U(f_{\epsilon E_{1}}x) = U([g_{\lambda}y]_{E_{2}}x)
		\end{equation}
		for $x \sim_{E_{1}} f_{\epsilon}$ and 
		\begin{equation}\label{ec2}
			U(f_{\epsilon E_{1}}x_{1}^{*}) = U([g_{\lambda}y]_{E_{2}}x_{2}^{*})
		\end{equation}
		for some sufficiently large $x_{1}^{*}$ and $x_{2}^{*}$ with $x_{E_{1}}x_{1}^{*} \sim x_{E_{2}}x_{2}^{*}$. 
		
		As $x_{1}^{*}$ and $x_{2}^{*}$  can be chosen after $f_{\epsilon}$ and $g_{\lambda}y$ are pinned down, thus both $U(f_{\epsilon E_{1}}x_{1}^{*}) $ and $U([g_{\lambda}y]_{E_{2}}x_{2}^{*})$ can be guaranteed to be evaluated at an extreme point of $C^{*}(E_{1})$ and an extreme point of $C^{*}(E_{2})$ respectively. 
		
		Then from equation (\ref{ec2}) one can further derive
		\begin{align*}
			\min\limits_{p\in C} \int_{\Omega} u(f_{\epsilon E_{1}}x^{*}_{1}) dp& = \min\limits_{p\in C} \int_{\Omega} u([g_{\lambda}y]_{E}x_{2}^{*}) dp\\
			\Rightarrow \min\limits_{p\in C^{*}(E_{1})}  \int_{E_{1}} u(f_{\epsilon}) dp + u(x^{*}_{1}) (1- p^{*}(E_{1})) &= \min\limits_{p\in C^{*}(E_{2})}  \int_{E_{2}} u(g_{\lambda}y) dp + u(x_{2}^{*}) (1- p^{*}(E_{2})) 
		\end{align*}
		
		Let $q_{f}$ and $q_{g}$ be the two extreme points in the set of posteriors of $C^{*}(E_{1})$ and $C^{*}(E_{2})$ that evaluate $f_{\epsilon}$ and $g_{\lambda}y$ respectively, then the last equality can be written as: 
		\begin{equation*}
			u(f_{\epsilon})\cdot q_{f} + u(x^{*}_{1}) (1- p^{*}(E_{1})) = u(g_{\lambda}y) \cdot q_{g} + u(x_{2}^{*}) (1- p^{*}(E_{2})) 
		\end{equation*}
		Furthermore as the condition $x_{E_{1}}x_{1}^{*} \sim x_{E_{2}}x_{2}^{*}$ implies that
		\begin{equation*}
			u(x) p^{*}(E_{1}) + u(x_{1}^{*})(1-p^{*}(E_{1}))  = u(x) p^{*}(E_{2}) + u(x_{2}^{*}) (1-p^{*}(E_{2})) 
		\end{equation*}
		i.e. 
		\begin{equation*}
			u(x_{1}^{*})(1-p^{*}(E_{1}))   -  u(x_{2}^{*}) (1-p^{*}(E_{2}))  = u(x)[p^{*}(E_{2}) - p^{*}(E_{1})]
		\end{equation*}
		As $E_{1}$ and $E_{2}$ are chosen arbitrarily, without loss of generality, let $p^{*}(E_{2}) - p^{*}(E_{1}) \geq 0$ and also notice that $u(x) \geq 0$. In the following, let $M \equiv u(x)[p^{*}(E_{2}) - p^{*}(E_{1})]$. Notice that $M \in [0, \epsilon]$. 
		
		Then equation (\ref{ec1}) is further equivalent to 
		\begin{equation}
			u(f_{\epsilon})\cdot q_{f} + M = \lambda u(g)\cdot q_{g} + (1-\lambda)u(y) 
		\end{equation}
		Therefore, for each $\lambda \in (0.1]$, equation (\ref{ec2}) holds if $(1-\lambda)u(y)$ is given by the following
		\begin{equation}\label{equ5}
			(1-\lambda)u(y) = u(f_{\epsilon})\cdot q_{f} + M - \lambda u(g)\cdot q_{g} 
		\end{equation}
		and $y$ is arbitrary if $\lambda = 1$. 
		
		Next consider the equation (\ref{ec1})
		\begin{equation*}
			\min\limits_{p \in C} \int_{\Omega} u(f_{\epsilon E_{1}}x)dp = \min\limits_{p\in C} \left\{ \int_{E_{2}}\lambda u(g)dp + (1-\lambda)u(y)p(E_{2}) + u(x)(1-p(E_{2}))    \right\}
		\end{equation*}
		Plugging into $(1-\lambda)u(y)$ from equation (\ref{equ5}) to the RHS and denote it by $R(\lambda)$: 
		\begin{align*}
			R(\lambda) = &\min\limits_{p\in C} \left\{ \int_{E_{2}}\lambda u(g)dp +u(f_{\epsilon})\cdot q_{f}p(E_{2}) + Mp(E_{2}) - \lambda u(g)\cdot q_{g} p(E_{2}) + u(x)(1-p(E_{2}))    \right\}\\
			& = \min\limits_{p\in C} \left\{ u(f_{\epsilon})\cdot q_{f}p(E_{2}) + u(x)(1-p(E_{2}))  + M  p(E_{2})+ \lambda \left[\int_{E_{2}} u(g)dp -u(g)\cdot q_{g} p(E_{2}) \right] \right\}
		\end{align*}
		Again, $R(\lambda)$ is a continuous function of $\lambda$ and the LHS of equation (\ref{ec1}) is a constant of $\lambda$. Thus it suffices to show $R(0) > $ LHS and $R(1) <$ LHS. 
		
		For $R(0)$ one has, 
		\begin{align*}
			R(0) & = \min\limits_{p\in C} \left\{ u(f_{\epsilon})\cdot q_{f}p(E_{2}) + u(x)(1-p(E_{2}))  + M  p(E_{2}) \right\} \\
			& \geq \min\limits_{p\in C} \left\{ u(f_{\epsilon})\cdot q_{f}p(E_{2}) + u(x)(1-p(E_{2})) \right\} \\ 
			& \geq  \min\limits_{p\in C} \left\{ u(x) p(E_{2}) + u(x)(1-p(E_{2})) \right\}  \\
			& = u(x) \\
			& > \min\limits_{p \in C} \int_{\Omega} u(f_{\epsilon E_{1}}x)dp  = \text{LHS}
		\end{align*}
		
		On the other hand for $R(1)$, 
		\begin{align*}
			R(1) & =  \min\limits_{p\in C} \left\{ u(f_{\epsilon})\cdot q_{f}p(E_{2}) + u(x)(1-p(E_{2}))  + M  p(E_{2})+ \left[\int_{E_{2}} u(g)dp -u(g)\cdot q_{g} p(E_{2}) \right] \right\}\\
			& \leq 2 \epsilon + \min\limits_{p\in C} \left\{\int_{E_{2}} u(g)dp -u(g)\cdot q_{g} p(E_{2})\right\}
		\end{align*}
		For the second term as $\alpha[E_{2}]$ is unique, its minimum is negative. Therefore, it suffices to find $\epsilon > 0$ such that 
		\begin{equation*}
			\epsilon < -\frac{1}{2} \min\limits_{p\in C} \left\{\int_{E_{2}} u(g)dp -u(g)\cdot q_{g} p(E_{2})\right\} 
		\end{equation*}
		Then given this $\epsilon$ one has 
		\begin{equation*}
			R(1) < 2 \epsilon - 2 \epsilon = 0  \leq \text{LHS}
		\end{equation*}
		Therefore, the existence of $\lambda \in (0,1)$ that solves equation  (\ref{equ5}) is guaranteed. 
		
		Finally, once $\lambda$ is solved, $u(y)$ is given by equation (\ref{equ5}): 
		\begin{equation*}
			u(y) = \frac{u(f_{\epsilon})\cdot q_{f} + M - \lambda u(g)\cdot q_{g}}{1-\lambda} \geq - \frac{\lambda}{1-\lambda} u(g)\cdot q_{g}
		\end{equation*}
		where the last inequality comes from $u(f_{\epsilon}(\omega)) + M \geq 0$. It remains to show that $u(y) \geq \underline{u}$ when $\underline{u}$ exists to guarantee the existence of this construction. It suffices to transform $g$ by taking mixtures with constant acts such that $u(g)\cdot q_{g} \leq 0$.  \\
		
		From this point on, apply exactly the same argument in step 3 of the proof of Theorem \ref{thm2} would imply that $\alpha[E_{1}] = \alpha[E_{2}]$. Therefore, $\alpha[E]$ needs to be a constant across all strict $\succsim$-nonnull $E\in \Sigma$. 
		\qed
		
		\newpage
		\section{On Assumption \ref{asum}}\label{apx1}
		In the main text, Assumption \ref{asum} is made for the ease of exposition. This appendix provides characterizations of ML without this assumption. The characterization results for Contingent RML and RML can be extended similarly, thus omitted in the present paper. 
		
		\subsection{ML with Bounded Ex-ante Preference}
		First consider the case where $u(X)$ is bounded. Without loss of generality, one can now normalize $u(X)$ to $[0,M]$ for some $ M \in \R_{++}$. Let $\bar{x}$ denote the best consequence. It should now be clear from the proof of Theorem \ref{thm1} that the key for characterization is to find $x^{*}$ such that $f_{E}x^{*}$ is ex-ante evaluated at some extreme point in $C^{*}(E)$. When $u(X)$ is bounded, the needed $x^{*}$ may not be available, then one can equivalently consider ``shrinking'' the act $f$. 
		
		For any act $f\in \mathcal{F}$ and $K \in  [1, \infty)$, let $f/K$ denote an act such that $u(f/K(\omega)) = u(f(\omega))/K$. Let $K_{E,f}$ be the threshold such that for all $K \geq K_{E,f}$, $(f/K)_{E}\bar{x}$ is evaluated at an extreme point in $C^{*}(E)$. Then the following axiom can be easily shown to be equivalent to ML in the current setting. \\
		
		\textbf{Axiom CR-S'} (\textbf{C}ontingent \textbf{R}easoning given \textbf{S}ufficiently large shrinking)
		
		For all $f\in \mathcal{F}$ and $x \in X$, if $f \sim_{E}x$, then $(f/K_{E,f})_{E}\bar{x} \sim (x/K_{E,f})_{E}\bar{x}$. \\
		
		On the other hand, instead of shrinking the act $f$, one can also consider taking mixtures between $f_{E}\bar{x}$ and $\underline{x}_{E}\bar{x}$. Similarly, one can find sufficiently small $\lambda_{E,f} \in (0,1)$ such that $\lambda_{E,f} f_{E}\bar{x} + (1-\lambda_{E,f}) \underline{x}_{E}\bar{x}$ is evaluated at an extreme point in $C^{*}(E)$.\footnote{I thank Rui Tang for suggesting this.} Then one can equivalently consider the following axiom: \\
		
		\textbf{Axiom CR-S''} (\textbf{C}ontingent \textbf{R}easoning given \textbf{S}ufficiently small mixing)
		
		For all $f\in \mathcal{F}$ and $x \in X$, if $f \sim_{E}x$, then $\lambda_{E,f} f_{E}\bar{x} + (1-\lambda_{E,f}) \underline{x}_{E}\bar{x} \sim \lambda_{E,f} x_{E}\bar{x} + (1-\lambda_{E,f}) \underline{x}_{E}\bar{x}$. \\
		
		Notice that, the existence of $K_{E,f}$ and $\lambda_{E,f}$ still depends on the finitely many extreme points assumption. The next subsection shows a way to relax it. 
		
		\subsection{ML with Infinitely Many Extreme Points}
		Suppose the set of priors $C$ could have infinitely many extreme points. Then a finite threshold for sufficiently good consequences may not exist. Instead one can consider a sequences of thresholds. Consider the following axiom:\\
		
		\textbf{Axiom Approximate CR-S} (Approximate \textbf{C}ontingent \textbf{R}easoning given \textbf{S}ufficiently good consequence).
		
		For all $f\in \mathcal{F}$ and for all $x, z, w \in X$ with $z \succ w$, there exists $\bar{x}_{E,f,z,w}$ such that for all $x^{*} \in X$ with $x^{*} \succsim \bar{x}_{E,f,z,w}$, if $f \sim_{E}x$,  then 
		\begin{equation*}
			\frac{1}{2} f_{E}x^{*} + \frac{1}{2} w \prec \frac{1}{2} x_{E}x^{*} + \frac{1}{2} z
		\end{equation*}
		and 
		\begin{equation*}
			\frac{1}{2} f_{E}x^{*} + \frac{1}{2}z \succ \frac{1}{2} x_{E}x^{*} + \frac{1}{2}w
		\end{equation*}
		
		Notice that CR-S implies Approximate CR-S. On the other hand, in the case where CR-S is silent as $\bar{x}_{E,f}$ does not exist, Approximate CR-S imposes an additional restriction on the behaviors. It restricts that the difference between $f_{E}x^{*}$ and $x_{E}x^{*}$ should be arbitrarily small when $x^{*}$ is sufficiently good. 
		
		Hence, Approximate CR-S conveys essentially the same intuition as CR-S. The following representation theorem shows that Approximate CR-S is equivalent to ML in the current setting. 
		
		\begin{theorem}\label{thm5}
			$\{\succsim_{E}\}_{E\in \Sigma}$ is represented by ML if and only if Approximate CR-S holds for all strict $\succsim$-nonnull events $E$. 
		\end{theorem}
		
		\begin{proof}[Proof of Theorem \ref{thm5}] 
			First consider the following lemma: 
			\begin{lemma}\label{alem}
				For any strict $\succsim$-nonnull $E\in \Sigma$ and $f\in \mathcal{F}$, for any $\epsilon > 0$ there exists $\bar{x}_{E,f,\epsilon} \in X$ such that
				\begin{equation*}
					\min\limits_{p\in C^{*}(E)} \int_{\Omega} u(f_{E}x^{*})dp -  \min\limits_{p\in C}\int_{\Omega} u(f_{E}x^{*}) dp  < \epsilon
				\end{equation*}
				for all $x^{*} \succsim \bar{x}_{E,f,\epsilon}$. 
			\end{lemma}
			
			\begin{proof}[Proof of Lemma \ref{alem}]
				For any strict $\succsim$-nonnull $E\in \Sigma$ and $f\in \mathcal{F}$, either there exists $\bar{x}_{E,f}$ such that 
				\begin{equation*}
					\min\limits_{p\in C} \int_{\Omega} u(f_{E}x^{*}) dp = \min\limits_{p\in C^{*}(E)} \int_{\Omega} u(f_{E}x^{*}) dp
				\end{equation*}
				for all $x^{*} \succsim \bar{x}_{E,f}$ or not. If it is the first case, then this lemma is trivially true. 
				
				Consider the case there does not exist $\bar{x}_{E,f}$ for some $E$ and $f$. For each $x \in X$, let $p_{x}$ be the probability measure in $C$ that evaluates the act $f_{E}x$ according to MEU, i.e. $p_{x} \equiv \arg\min_{p\in C} \int_{\Omega} u(f_{E}x)dp $. Let $q$ denote the probability measure in $C^{*}(E)$ that evaluates the act $f_{E}x$ and notice that it does not depend on the value of $x$. 
				
				Then one has 
				\begin{align*}
					&\quad \min\limits_{p\in C^{*}(E)} \int_{\Omega} u(f_{E}x)dp -  \min\limits_{p\in C}\int_{\Omega} u(f_{E}x) dp \\
					& = \int_{E} u(f)dq + u(x)(1-p^{*}(E)) - \int_{E} u(f)dp_{x} - u(x)(1-p_{x}(E))
				\end{align*}
				
				Take derivative with respect to $u(x)$ and apply envelope theorem yields
				\begin{equation*}
					\frac{d}{du(x)} \left[ \min\limits_{p\in C^{*}(E)} \int_{\Omega} u(f_{E}x)dp -  \min\limits_{p\in C}\int_{\Omega} u(f_{E}x) dp \right] = p_{x}(E) - p^{*}(E)
				\end{equation*}
				
				The current assumption $p_{x} \notin C^{*}(E)$ implies that $p_{x}(E) - p^{*}(E) < 0$, i.e. the difference is decreasing with respect to $u(x)$. Furthermore, since the difference is bounded below by zero, monotone convergence theorem implies that
				\begin{equation*}
					\min\limits_{p\in C^{*}(E)} \int_{\Omega} u(f_{E}x)dp -  \min\limits_{p\in C}\int_{\Omega} u(f_{E}x) dp \rightarrow 0 
				\end{equation*}
				for $u(x) \rightarrow \infty$, i.e. the lemma holds. 
				
			\end{proof}
			
			For the necessity of Approximate CR-S, recall under ML updating, $f\sim_{E}x$ implies
			\begin{equation*}
				\min\limits_{p\in C^{*}(E)} \int_{\Omega} u(f_{E}x^{*})dp = \min\limits_{p\in C} \int_{\Omega} u(x_{E}x^{*})dp 
			\end{equation*}
			Thus for any $\epsilon > 0$ there exists $\bar{x}_{E,f,\epsilon}$ such that 
			\begin{equation*}
				\min\limits_{p\in C} \int_{\Omega} u(x_{E}x^{*})dp - \min\limits_{p\in C}\int_{\Omega} u(f_{E}x^{*}) dp  < \epsilon
			\end{equation*}
			for all $x^{*} \succsim \bar{x}_{E,f,\epsilon}$. Then for each $z \succ w$, it suffices to let $\epsilon = u(z) - u(w)$ and then Approximate CR-S axiom holds. \\
			
			For sufficiency, fix any strict $\succsim$-nonnull $E$ and consider the contrapositive statement: not ML updating implies not Approximate CR-S axiom. Not ML means that $C_{E} \neq \{ p/p(E): p \in C^{*}(E) \}$. In other words, either there exists $\tilde{p} \in C^{*}(E)$ such that $\tilde{p}/\tilde{p}(E) \notin C_{E}$ or there exists $q \in C_{E}$ such that $q \notin \{p/p(E): p \in C^{*}(E)\}$ or both. 
			
			Consider the first case, by the same argument implied by the strict separating hyperplane Theorem in the proof of Theorem \ref{thm1}, there exists $f\in \mathcal{F}$ and $x\in X$ such that $f\sim_{E}x$ and 
			\begin{equation*}
				\min\limits_{p\in C^{*}(E)} \int_{E}u(f)\frac{dp}{p(E)} < u(x)
			\end{equation*}
			Then for all $x^{*} \succsim x$, 
			\begin{align*}
				\min\limits_{p\in C} \int_{\Omega} u(x_{E}x^{*})dp -  \min\limits_{p\in C^{*}(E)} \int_{\Omega} u(f_{E}x^{*})dp = u(x)p^{*}(E) - \min\limits_{p\in C^{*}(E)} \int_{E}u(f)dp> 0
			\end{align*}
			That is, there exists $\delta_{E,f} > 0 $ such that 
			\begin{align*}
				\min\limits_{p\in C} \int_{\Omega} u(x_{E}x^{*})dp -  \min\limits_{p\in C^{*}(E)} \int_{\Omega} u(f_{E}x^{*})dp > \delta_{E,f} > 0
			\end{align*}
			Now for all $x^{*} \succsim x$, 
			\begin{align*}
				&\quad \min\limits_{p\in C} \int_{\Omega} u(x_{E}x^{*})dp - \min\limits_{p\in C}\int_{\Omega} u(f_{E}x^{*})dp \\
				&=  \min\limits_{p\in C} \int_{\Omega} u(x_{E}x^{*})dp - \min\limits_{p\in C^{*}(E)} \int_{\Omega} u(f_{E}x^{*})dp  + \min\limits_{p\in C^{*}(E)} \int_{\Omega} u(f_{E}x^{*})dp  - \min\limits_{p\in C}\int_{\Omega} u(f_{E}x^{*})dp \\
				& > \delta_{E,f}  > 0
			\end{align*}
			the first inequality comes from the fact that $\min\limits_{p\in C^{*}(E)} \int_{\Omega} u(f_{E}x^{*})dp  - \min\limits_{p\in C}\int_{\Omega} u(f_{E}x^{*})dp \geq 0$. Then it suffices to find $z \succ w$ such that $u(z) - u(w) < \delta_{E,f}$ and it will imply that for all $x^{*} \succsim x$: 
			\begin{equation*}
				\frac{1}{2}f_{E}x^{*} + \frac{1}{2}z \prec \frac{1}{2} x_{E}x^{*} + \frac{1}{2} w
			\end{equation*}
			i.e. the Approximate CR-S axiom fails. \\
			
			Now consider the second case, there exists $q \in C_{E}$ such that $q \notin \{p/p(E): p \in C^{*}(E)\}$. By strict separating hyperplane theorem, there exists $f \sim_{E}x$ such that 
			\begin{equation*}
				u(x) = \min\limits_{p\in C_{E}} \int_{\Omega} u(f) dp \leq \int_{\Omega} u(f) dq < \min\limits_{p\in C^{*}(E)} \int_{E}u(f)\frac{dp}{p(E)}
			\end{equation*}
			Then it further implies for all $x^{*} \succsim x$, 
			\begin{equation*}
				\min\limits_{p\in C^{*}(E)} \int_{\Omega} u(f_{E}x^{*})dp  - \min\limits_{p\in C} \int_{\Omega} u(x_{E}x^{*})dp  =  \min\limits_{p\in C^{*}(E)} \int_{E}u(f)dp - u(x)p^{*}(E) > 0
			\end{equation*}
			i.e. there exists $\delta_{E,f}$ such that 
			\begin{equation*}
				\min\limits_{p\in C^{*}(E)} \int_{\Omega} u(f_{E}x^{*})dp  - \min\limits_{p\in C} \int_{\Omega} u(x_{E}x^{*})dp  > \delta_{E,f} > 0
			\end{equation*}
			On the other hand, by Lemma \ref{alem}, for any $\epsilon >0$ there exists $\bar{x}_{E,f,\epsilon}$ such that 
			\begin{equation*}
				\min\limits_{p\in C^{*}(E)} \int_{\Omega} u(f_{E}x^{*})dp - \min\limits_{p\in C}\int_{\Omega} u(f_{E}x^{*}) dp < \epsilon
			\end{equation*}
			for all $x^{*} \succsim \bar{x}_{E,f,\epsilon}$. Now for all $x^{*} \succsim \max\{x, \bar{x}_{E,f,\epsilon}\}$, 
			\begin{align*}
				&\quad  \min\limits_{p\in C}\int_{\Omega} u(f_{E}x^{*})dp - \min\limits_{p\in C} \int_{\Omega} u(x_{E}x^{*})dp \\
				&=  \min\limits_{p\in C}\int_{\Omega} u(f_{E}x^{*})dp - \min\limits_{p\in C^{*}(E)} \int_{\Omega} u(f_{E}x^{*})dp  + \min\limits_{p\in C^{*}(E)} \int_{\Omega} u(f_{E}x^{*})dp  - \min\limits_{p\in C} \int_{\Omega} u(x_{E}x^{*})dp\\
				& > -\epsilon + \delta_{E,f} 
			\end{align*}
			Now find any $z \succ w$ such that $u(z) - u(w) = \eta < \delta_{E,f}$ and let $\epsilon = \delta_{E,f} - \eta > 0$. Then for all $x^{*} \succsim \max\{x, \bar{x}_{E,f,\epsilon}\}$ the previous result implies that 
			\begin{equation*}
				\min\limits_{p\in C}\int_{\Omega} u(f_{E}x^{*})dp - \min\limits_{p\in C} \int_{\Omega} u(x_{E}x^{*})dp > -\epsilon + \delta_{E,f} = \eta > 0
			\end{equation*}
			Then 
			\begin{align*}
				\frac{1}{2}  \min\limits_{p\in C}\int_{\Omega} u(f_{E}x^{*})dp + \frac{1}{2}u(w) - \min\limits_{p\in C} \int_{\Omega} u(x_{E}x^{*})dp - \frac{1}{2}u(z) > \frac{1}{2}[\eta - \eta] = 0 
			\end{align*}
			i.e. $\frac{1}{2} f_{E}x^{*} + \frac{1}{2} w \succ \frac{1}{2}x_{E}x^{*} + \frac{1}{2}z$ for all $x^{*} \succsim \bar{x}_{E,f,\epsilon}$. Thus the Approximate CR-S axiom fails in this case as well, i.e. not ML implies not Approximate CR-S. 
		\end{proof}
		
		Finally, in the case where $u(X)$ is bounded and $C$ has infinitely many extreme points. It suffices to consider axiom Approximate CR-S' or Approximate CR-S'' that can be analogously defined. 
		
	\end{appendices}

	\newpage
	\bibliographystyle{econ}
	\bibliography{/Users/xiaoyucheng/Dropbox/Research/References/mylibrary.bib}

\end{document}